\documentclass[twocolumn,twoside]{IEEEtran}
\usepackage{times} 
\usepackage{epsfig}
\usepackage{graphicx}
\usepackage{epstopdf}
\usepackage{algorithm,algorithmic,color,subfigure, latexsym, amsmath, amsfonts, amssymb, cite}
\usepackage{algorithm}
\usepackage{algorithmic}
\usepackage{epsfig}
\usepackage{graphicx}
\usepackage{epstopdf}
\usepackage{bm,amsmath,amssymb,amsfonts,graphicx,epsfig,amsthm,color}
\usepackage{bbold,dsfont}
\usepackage{float}

\usepackage[hyphens]{url}
\PassOptionsToPackage{bookmarks=false}{hyperref}

\usepackage[depth=-1]{bookmark}

\usepackage{booktabs}       
\usepackage{amsfonts}       
\usepackage{microtype}      

\usepackage{times}
\usepackage{graphicx} 

\usepackage{wrapfig}



\usepackage{accents}
\makeatletter
\def\wid{\check{{\cc@style\underline{\mskip9.5mu}}}}
\def\Wideubar{\underaccent{{\cc@style\underline{\mskip6mu}}}}
\makeatother

\makeatletter
\def\wideubar{\underaccent{{\cc@style\underline{\mskip9.5mu}}}}
\def\Wideubar{\underaccent{{\cc@style\underline{\mskip6mu}}}}
\makeatother

\makeatletter
\def\widebar{\accentset{{\cc@style\underline{\mskip9.5mu}}}}
\def\Widebar{\accentset{{\cc@style\underline{\mskip6mu}}}}
\makeatother

\newtheorem{Lemma}{Lemma}

\newtheorem{Theorem}{Theorem}

\newtheorem{Assumption}{Assumption}
\newtheorem{Example}{Example}

\theoremstyle{Remark}\newtheorem{Remark}{Remark}

\allowdisplaybreaks
\begin{document}
\title{\LARGE \bf Data-driven Control of Dynamic Event-triggered Systems with  Delays}

\author{
	Xin Wang,~Jian Sun,
\IEEEmembership{Senior Member, IEEE},
Julian Berberich,
Gang Wang, \IEEEmembership{Member,~IEEE},
\\
Frank Allg{\"o}wer,
and~Jie Chen,
\IEEEmembership{Fellow,~IEEE}
\thanks{This work was supported in part by the National Key R$\&$D Program of China under Grant 2018YFB1700100, the National Natural Science Foundation of China under Grants 61925303, 62088101, {\color{blue}62173034,} U20B2073, 61720106011, and the Deutsche Forschungsgemeinschaft (DFG, German Research Foundation) {\color{blue}under Germany's Excellence Strategy - EXC 2075 - 390740016 and under grant 468094890. We acknowledge the support by the Stuttgart Center for Simulation Science (SimTech). The authors thank the International Max Planck Research School for Intelligent Systems (IMPRS-IS) for supporting J. Berberich.}}
\thanks{
X. Wang, J. Sun, and G. Wang are with the State Key Lab of Intelligent Control and Decision of Complex Systems, School of Automation, Beijing Institute of Technology, Beijing 100081, China, {\color{blue}and also with the Beijing Institute of Technology Chongqing Innovation Center, Chongqing 401120, China} (e-mail: xinwang@bit.edu.cn, sunjian@bit.edu.cn, gangwang@bit.edu.cn).


J. Berberich and F. Allg{\"o}wer are with the Institute for Systems Theory and Automatic Control, University of Stuttgart, 70550 Stuttgart, Germany (e-mail: julian.berberich@ist.uni-stuttgart.de, frank.allgower@ist.uni-stuttgart.de).


J. Chen is with the Department of Control Science and Engineering, Tongji University, Shanghai 201804, China, and also with the State Key Lab of Intelligent Control and Decision of Complex Systems, School of Automation, Beijing Institute of Technology, Beijing 100081, China 	
(e-mail: chenjie@bit.edu.cn).}
}

\maketitle

\begin{abstract}
This paper studies data-driven control of unknown sampled-data systems with communication delays under an event-triggering transmission mechanism. Data-based representations for time-invariant linear systems with known or unknown system input matrices are first developed, along with a novel class of dynamic triggering schemes for sampled-data systems with time delays. A model-based stability condition for the resulting event-triggered time-delay system is established using a looped-functional approach.
Combining this model-based condition with the data-driven representations, data-based stability conditions are derived.
Building on the data-based conditions, methods for co-designing the controller gain and the event-triggering matrix are subsequently provided for both cases with or without using the input matrix.
Finally, \textcolor{blue}{numerical examples are presented to corroborate the usefulness of additional prior knowledge of the input matrix in reducing the conservatism of stability conditions,} as well as the merits of the proposed data-driven event-triggered control schemes relative to existing results.
\end{abstract}


\section{Introduction}
\label{sec:introduction}
Sampled-data control has aroused great interests in the study of networked control systems due to its convenience in system design and analysis \cite{Chen1995,eng2022}.
A fundamental concept for sampled-data systems is so-called maximum sampling interval (MSI) for maintaining system stability.
In the literature, previous approaches for analyzing sampled-data systems and computing the MSIs can be generally categorized into four groups \cite{HETEL2017}, which include the discrete-time approach \cite{Fuj2009}, impulsive systems approach \cite{Nag2008}, input/output stability approach \cite{L2007}, and input-delay approach \cite{Fridman2010}.
On top of the discrete-time and input-delay approaches, a looped-functional approach was recently developed in \cite{Seuret2012,Briat2012}.
This proposal has been shown to yield (markedly) improved stability conditions by introducing a looped-functional into the common Lyapunov functional, 
 where the looped-functional is required to be continuous at sampling points but not necessarily to be positive definite; see, e.g., using a two-sided looped-functional in \cite{ZENG2017328},
an integral-based looped-functional in \cite{Lee20177}, and an extended looped-functional in \cite{Park2020}. In sampled-data systems, communication delays are natural, inevitable, and they have a considerable impact on the MSI \cite{Lee2020,Sun2010}.
Recent efforts in \cite{Lee2020,Zeng2019,Zhangck2014} have been devoted to analyzing sampled-data systems with communication delays by means of the looped-functional approach. 

It is worth stressing that all the aforementioned results are model-based, in the sense that they require explicit knowledge of the underlying system models.
Nonetheless, obtaining an accurate system model is challenging and may even be impossible in many real-world applications. \textcolor{blue}{
System identification \cite{Ljung1986} can be used to determine a model based on measured data. However, it is difficult to obtain tight error bounds in system identification when only noisy data of finite length are available \cite{Matni2019}. Therefore, as an alternative, direct data-driven control without prior identification has received increasing attention in recent years, see \cite{Hou2013} for a survey.}
A number of data-driven approaches have been put forth for different control problems using the celebrated Fundamental Lemma \cite{WILLEMS2005}, including state-feedback and optimal control in \cite{persis2020,van2020}, robust control in \cite{berberich2020robust,Waarde2020},
predictive control in \cite{Coulson2019,berberich2019a,liu2021data,Zhou2021sampling}.
An extension studying data-driven system stabilization of as well as robustness guarantees for systems with time-delays and measurement noise was recently investigated in \cite{rueda2020data}.
Moreover, data-driven approaches for estimating the MSI and designing sampled-data controllers under aperiodic sampling were developed in \cite{wildhagen2021datadriven} and \cite{wildhagen2021improved}.
Wedding the data-driven system representation with the time-delay approach, a data-based stability condition for \emph{continuous-time} sampled-data control systems was derived in \cite{Berberich2020}, along with a data-driven controller.
Leveraging the stability condition, MSIs under different levels of noise can be computed.
Nonetheless, these results are conservative due to the fact that the data-based condition in \cite{Berberich2020} follows from the model-based one in \cite{Fridman2010}.
Improved MSIs can be expected if less conservative model-based stability conditions are employed through looped-functional approach. 

In real sampled-data systems, the digital communication network has limited bandwidth. Reducing ``redundant" data transmissions can be  helpful for saving communication resources while maintaining desired control performance.
To this aim, an event-triggering transmission strategy was proposed in \cite{Tabuada2007,Chen2021event}.
Its key idea is to monitor the system continuously, but transmit sampled data only when ``necessary'' according to some predefined criterion. Following this idea, a multitude of sampling-based event-triggering schemes were devised, including e.g., discrete event-trigger \cite{Yue2013}, dynamic event-trigger \cite{Girard2015}, dynamic periodic event-trigger \cite{Liu2018}. From an energy equipartition perspective,  energy-event-triggered control of cyber-physical network systems was developed in \cite{zeng2015energy} for robust and fast stabilization.
\textcolor{blue}{However,
these existing event-triggered control methods all require explicit system models.
In particular, they cannot handle uncertainty which may arise from an identification step in case only data are available.
To the best of the authors' knowledge,
it remains an untapped field to design data-driven event-triggering schemes as well as controllers for unknown sampled-data control systems from noisy data with rigorous stability guarantees.}

These recent advances and developments have motivated our work in this paper, which is focused on data-driven event-triggered control of sample-data systems with communication delays.
\textcolor{blue}{Motivated by dynamic periodic event-triggering from \cite{Liu2018}, we develop a discrete-time dynamic periodic event-triggering scheme that correlates with communication delays to save transmission resources.}
Subsequently, model- and data-based stability conditions for sampled-data systems with delays under the triggering scheme are obtained by employing the looped-functional approach.
\textcolor{blue}{
Based on the model-based criteria, data-driven co-design methods of the controller gain and the triggering matrix are given for both cases having known or unknown input matrices.
Possibly large MSIs can be obtained, similar to \cite{Berberich2020}, by iteratively solving two semi-definite programs which relax the underlying non-convex design problem. 
Alternative methods using an algebraically equivalent systems \cite{ZHANG201555} are then provided to avoid such iterative computation, where the design conditions are convex in all decision variables.}
 \textcolor{blue}{Furthermore, it has been shown in \cite{berberich2020combining} that leveraging prior system knowledge could provide less conservative system analysis and design results than purely data-driven approaches.}
 Complementing the system representation in \cite{Berberich2020}, a new data-dependent parametrization with known input matrices is established.
It turns out that incorporating the prior knowledge of the input matrix can help to further reduce conservatism of the resulting data-based stability conditions, which is confirmed by our numerical results.

In words, the main contributions of this work are summarized as follows.
{\color{blue}
\begin{enumerate}
\item [\textbf{c1)}] A novel dynamic event-triggering scheme for sampled-data systems with delays, where the dynamic threshold variable only changes at sampling points, eliminating the burden of continuously computing the dynamic variable in \cite{Girard2015,Liu2018};
\item [\textbf{c2)}] Model- and data-based stability conditions in the form of linear matrix inequalities (LMIs) for event-triggered control systems with delays based on a novel looped-functional; and,
\item [\textbf{c3)}] Data-driven approaches to co-designing the controller gain and the triggering matrix based on a data-based system parametrization for both cases with known and unknown input matrices. 
\end{enumerate}}
Practical merits and effectiveness of the proposed data-driven event-triggered control as well as stability conditions are finally demonstrated through numerical comparison using two simulated examples.

The remainder of this paper is structured as follows. In Section \ref{Sec:preliminaries}, we introduce data-based system representations, as well as put forth a dynamic event-triggering scheme for sampled-data systems with noise and time delays. Based on the results in Section \ref{Sec:preliminaries}, both model- and data-based stability conditions are developed in Section \ref{sec:mainresults}.
Subsequently, methods for co-designing the controller gain and the triggering matrix are discussed. Section \ref{sec:example} validates the merits and application of our methods and conditions on practical systems. Finally, Section \ref{sec:conclude} draws concluding remarks.

{\it Notation.}
Throughout this paper, $\mathbb{N}$, $\mathbb{R^+}$, $\mathbb{R}^n$, and $\mathbb{R}^{n\times m}$ denote the sets of all non-negative
integers, non-negative real numbers, $n$-dimensional vectors, and ${n\times m}$ real matrices, respectively.
Symbol $P\succ 0$ ($P\succeq 0$) means $P$ is a symmetric positive (semi)definite matrix;
${\rm diag}\{\cdots\}$ denotes a block-diagonal matrix;
${\rm Sym}\{P\}$ represents the sum of $P^{T}$ and $P$. Let
${\rm col} \{\cdots\}$ represent column vectors; $0$ and $I$ stand for zero and identity matrices
of appropriate
dimensions.
Notation `$\ast$' denotes the symmetric term in symmetric block matrices;
We use $\|\cdot\|$ to denote the Euclidean norm of a vector.
$t_{k}^-$ means that $t$ approaches $t_k$ from left. {\color{blue}$\{\tau_j-d\}_{j}$ represents the set $\{\tau_0-d, \tau_1-d, \tau_2-d, \dots\}$. We write $[\cdot]$ if elements in the matrix can be inferred by symmetry.}

\section{Preliminaries}\label{Sec:preliminaries}
\subsection{Sampled-data control with unknown system matrices}


Consider the following linear time-invariant system 
\begin{equation}\label{sys:LTI}
\dot{x}(t)=A x(t)+B u(t)
\end{equation}
where $t\geq0$, $ x(t)\in \mathbb{R}^{n}$ is the system state, and $u(t)\in \mathbb{R}^{m}$ is the control input. $A\in \mathbb{R}^{n\times n}$ and $B\in \mathbb{R}^{n\times m}$ are constant system matrices.
An {\emph {unknown}} \textcolor{blue}{constant} time-delay $d$ in the network communication is considered and it satisfies $0\leq \underline{d}\leq d\leq\bar{d}$, where the bounds $\underline{d}$ and $\bar{d}$ are given constants.
We assume that some pre-collected
state-input measurements are available.
Then two different assumptions on system matrices $A$ and $B$ are considered in this paper, which are formalized as follows.
\begin{Assumption}\label{Ass:matrix:AB}
The matrices $A$ and $B$ are unknown.
\end{Assumption}
\begin{Assumption}\label{Ass:matrix:A}
The matrix $A$ is unknown, and $B$ is known.
\end{Assumption}
Our main goal is to design a fixed control gain matrix $K\in\mathbb{R}^{m\times n}$ such that the state-feedback control signal $u(t)=Kx(t-d)$
stabilizes system \eqref{sys:LTI} under Assumption \ref{Ass:matrix:AB} or \ref{Ass:matrix:A}.
Suppose that the control signals are released at discrete time instants
\begin{align}\label{sys:instants}
0\leq d= t_0<t_1<\dots<t_k<\dots<\infty.
\end{align}
The sampled-data system in closed loop can be written as
\begin{equation}\label{sys:sampling}
\left\{
\begin{array}{lll}
\dot{x}(t)=A x(t)&, &t\in[0,t_0)\\
\dot{x}(t)=A x(t)+BKx(t_k-d)&, &t\in[t_k,t_{k+1}).\\
x(t)=0&, &t\in[-d,0)
\end{array}
\right.
\end{equation}

A framework for system \eqref{sys:sampling} is presented in Fig. \ref{FIG:structure}. Two differences relative to classic sampled-data systems can be noted: i) including a dynamic event-triggering transmission module to determine the transmission sequence $\{t_k\}_{k\in \mathbb{N}}$, and ii) a data-based system representation for co-designing the event-triggering scheme and the controller. {\color{blue}The data-based system representation is constructed based on the data collected offline at discrete instants $\{T_i\}_i$. These sampling times are independent of the closed-loop sampling times $\{t_k\}_k$ at which the controller is updated.
Next, we elaborate on the data-based representation and the designed event-triggering transmission scheme, separately.}
\begin{figure}
	\centering	 \includegraphics[width=\columnwidth]{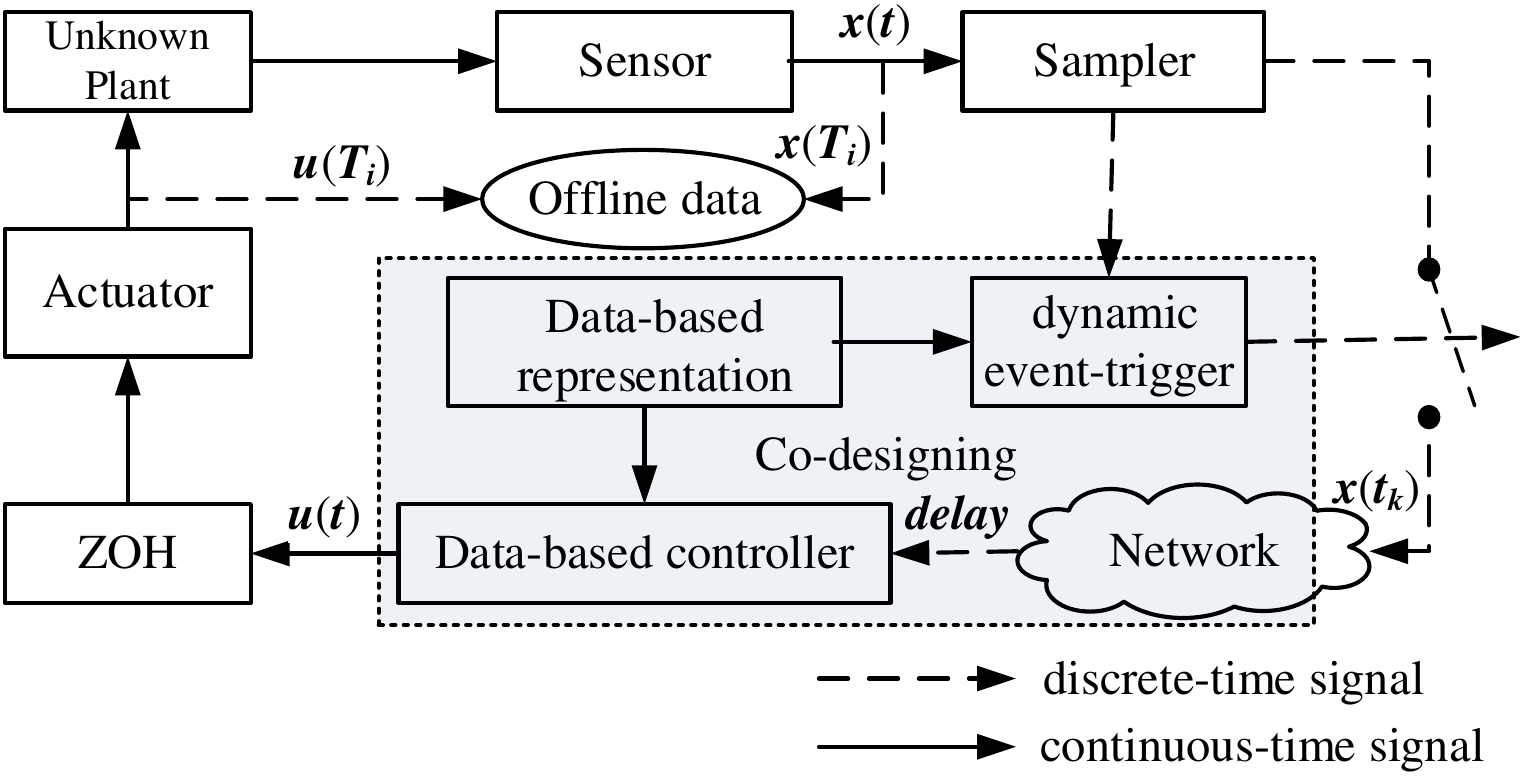}
	\caption{\color{blue}Graphical illustration of the proposed approach.}
	\label{FIG:structure}
\end{figure}
\subsection{Data-based system representation with noise}

Under Assumption \ref{Ass:matrix:AB}, a main challenge is that both system stability analysis and controller design should be performed without knowledge of the system matrices $A$ and $B$.
To this end, we recall the data-based representation developed in \cite{Waarde2020,Berberich2020} for system \eqref{sys:LTI}.
\textcolor{blue}{
Suppose that sampled measurements $\{\dot x(T_i),\,x(T_i),\,u(T_i)\}^{\rho}_{i=1}$ $(i,\,\rho \in \mathbb{N})$ of the perturbed system}
\begin{equation}\label{sys:data:perturbed}
\dot{x}(t)=A x(t)+B u(t)+B_ww(t)
\end{equation}
are available at discrete time instants $T_i \in [0,t]$, where $B_w\in \mathbb{R}^{n\times n_w}$ is assumed known and has full column rank. This is without loss of generality since one can also lump $B_{w}$ into the additive noise $w(t)$.
Here,
the measured system data are corrupted by the {\it unknown} noise (perturbation) sequence $\{w(T_i)\}^{\rho}_{i=1}$, where $w(t)\in \mathbb{R}^{n_w}$ captures the unknown noise or the unmodeled system dynamics.
These measurements can be collected to form the following data matrices
\begin{align*}
\dot X&:=[\dot x(T_1)~~\dot x(T_2)~~\dots~~ \dot x(T_\rho) ]\\
X&:=[x(T_1)~~x(T_2)~~\dots~~ x(T_\rho) ]\\
U&:=[u(T_1)~~u(T_2)~~\dots~~ u(T_\rho) ]\\
W&:=[w(T_1)~~w(T_2)~~\dots~~ w(T_\rho) ]
\end{align*}
where $\dot{X}$, $X$, and $U$ are available while $W$ is unknown. Then,
it is evident that
\begin{align}\label{formu:data}
\dot X=AX+BU+B_wW.
\end{align}
In practice, the noise is typically bounded.
We make a standing assumption on the noise that was also used in  {\cite{Waarde2020}} and \cite{Berberich2020}.
\begin{Assumption}[\emph {Noise bound}]\label{Ass:disturbance}
The noise sequence $\{w(T_i)\}^{\rho}_{i=1}$ collected in the matrix $W$ satisfies $W\in\mathcal{W}$ with
\begin{align}\label{data:disturbance}
\mathcal{W}=\bigg\{W\in\mathbb{R}^{n_w\times\rho} \Big |
\left[\!\begin{array}{cc}W^T\\I \\\end{array}\!\right]^T
  \left[\!\begin{array}{cc}Q_d\! & \!S_d\\\ast\! & \!R_d \\\end{array}\!\right]
  \left[\!\begin{array}{cc}W^T\\I \\\end{array}\!\right]\succeq0 \bigg\}
\end{align}
for some known matrices $Q_d \prec 0 \in \mathbb{R}^{\rho\times \rho}$, $S_d \in \mathbb{R}^{\rho\times n_w}$, and $R_d=R_d^T \in \mathbb{R}^{n_w\times n_w}$.
\end{Assumption}

\begin{Remark}[\emph {Explanation of noise}]\label{remark:noise}
Assumption \ref{Ass:disturbance} provides a general form for modeling bounded additive noise, which has been used in similar form by \cite{persis2020,Berberich2020,Waarde2020,rueda2020data}. {\color{blue}The matrices $Q_d$, $S_d$, $R_d$ are chosen depending on the prior knowledge on the noise.
For example, if $w(t)\in [-\bar{w},\bar{w}]$ for $\bar{w}\geq0$, \eqref{data:disturbance} boils down to
\begin{align}
WW^T=\sum \limits_{i=1}^{i=\rho} w(T_i) w^T(T_i)\preceq \bar{w}^2\rho I
\end{align}
where $Q_d=-I$, $R_d=0$, and $R_d=\bar{w}^2\rho I$. The noise bound in Assumption \ref{data:disturbance} has been employed in many recent works on data-driven control, and allows for various extensions (e.g.,
using multipliers instead of constant matrices $Q_d$, $S_d$, and $R_d$ in \cite{berberich2020combining}).
}
\end{Remark}
Under the model in \eqref{formu:data} and Assumption \ref{Ass:disturbance}, we define $\Sigma_s$ as the set of all pairs $(A~B)$ adhering to the measured data and the noise bound, namely
\begin{align}\label{sys:data}
\Sigma_s:=\{(A~ B)\mid \dot X=AX+BU+B_wW,~ W\in \mathcal{W}\}.
\end{align}

Through a simple mathematical replacement method, an equivalent expression of $\Sigma_s$ in the form of a quadratic matrix inequality (cf. \cite[Theorem 1]{Berberich2020})
is given below.
\begin{Lemma}[\emph {Data-based representation}] \label{Lemma:system:data}
The set $\Sigma_s$ in \eqref{sys:data} is equivalent  to
\begin{align}\label{data:represent}
\Sigma_s=\bigg\{[A~B]\in\mathbb{R}^{\color{blue}n\times (n+m)} \Big |
\left[\!\begin{array}{cc}[A~B]^T\\I \\\end{array}\!\!\!\right]^T
  \!\Theta_s\!
  \left[\!\begin{array}{cc}[A~B]^T \\
   I\\\end{array}\!\!\!\right]\succeq0
\bigg\}
\end{align}
where
\begin{align}
\Theta_s=\left[\!\begin{array}{cc}Q_c\! & \!S_c\\\ast\! & \!R_c \\\end{array}\!\right]:=\left[\!\begin{array}{cc}-\mathcal{Z}\! & \!0\\ \dot X\! &\!B_w \\\end{array}\!\right]
  \left[\!\begin{array}{cc}Q_d\! & \!S_d\\\ast\! & \!R_d \\\end{array}\!\right]
  \left[\!\begin{array}{cc}-\mathcal{Z}\! & \!0\\ \dot X\! &\!B_w\\\end{array}\!\right]^T
  \end{align}
  with $\mathcal{Z}:={\rm col}\{X, \,U\}$.
\end{Lemma}

For the sake of stability analysis in subsequent sections,
a technical assumption similar to {\cite[Assumption 3]{Berberich2020}} on the matrix $\Theta_s$ in Lemma \ref{Lemma:system:data} is made below.
\begin{Assumption}\label{Ass:matrix}
The matrix $\Theta_s$ is invertible and has $n_w$ positive eigenvalues.
\end{Assumption}

Lemma \ref{Lemma:system:data} provides a purely data-based representation of system \eqref{sys:LTI} with unknown matrices $A$ and $B$ (cf. Assumption \ref{Ass:matrix:AB}) using noisy data $\dot X$, $X$,  and $U$.
{If additional prior knowledge on the system is given, then conservatism can be reduced in the co-design step of controller and triggering mechanism. Specifically, we consider the case of knowing the input matrix $B$ (cf. Assumption \ref{Ass:matrix:A}), which allows us to state significantly less conservative design conditions.
To treat the data-driven control under Assumption \ref{Ass:matrix:A},
the set of all matrices $A$ consistent with the model in \eqref{formu:data} and the noise bound in Assumption \ref{Ass:disturbance} is defined as follows.}
\begin{align}\label{sys:data:k}
\bar{\Sigma}_s:=\{A\in\mathbb{R}^{n\times n} \mid \dot X=AX+BU+B_wW,~W\in \mathcal{W}\}.
\end{align}
Similar to Lemma \ref{Lemma:system:data}, an equivalent representation for $\bar{\Sigma}_s$ is given in the following.

\begin{Lemma}[\it{Data-based representation with known input matrix $B$}]\label{Lemma:system:data:K}
The set in \eqref{sys:data:k} is equal to
\begin{align}\label{data:represent:K}
\bar{\Sigma}_s=\bigg\{A\in\mathbb{R}^{n\times n}  \Big |
\left[\begin{array}{cc}A^T\\I \\\end{array}\right]^T
  \bar{\Theta}_s
  \left[\begin{array}{cc}A^T\\I \\\end{array}\right]\succeq0
\bigg\}
\end{align}
where
\begin{align}
\bar{\Theta}_s=\left[\!\begin{array}{cc}\bar{Q}_c\! & \!\bar{S}_c\\\ast\! & \!\bar{R}_c \\\end{array}\!\right]:=\left[\!\begin{array}{cc}-X\! & \!0\\ \mathcal{U}\! &\!B_w \\\end{array}\!\right]
  \left[\!\begin{array}{cc}Q_d\! & \!S_d\\\ast\! & \!R_d \\\end{array}\!\right]
  \left[\!\begin{array}{cc}-X\! & \!0\\ \mathcal{U}\! &\!B_w\\\end{array}\!\right]^T
  \end{align}
  with $\mathcal{U}:=\dot X-BU$.
\end{Lemma}

{Lemma \ref{Lemma:system:data:K}
	presents a data-based representation of system \eqref{sys:LTI} while accounting for prior knowledge of the input matrix $B$. 
As discussed in \cite{berberich2020combining}, additional prior knowledge of the system may help shrink the set of systems consistent with the collected data, thereby contributing to less conservative data-based stability conditions compared to the ones derived using only data.
This can be understood as a data- and model-driven hybrid approach.
In fact, it can be shown that the set $\bar{\Sigma}_s$ in Lemma  \ref{Lemma:system:data:K} is tighter than $\Sigma_s$ in Lemma  \ref{Lemma:system:data} for system \eqref{sys:LTI}.
On the other hand, prior knowledge of the input matrix $B$ may not always be available in general.
These considerations thus motivates us to pursue data-driven control schemes based on $(i)$ pure data as well as $(ii)$ data along with prior knowledge of $B$, which are detailed in Section \ref{sec:mainresults}.}

\begin{Assumption}\label{Ass:matrix:k}
The matrix $\bar{\Theta}_s$ is invertible and has $n_w$ positive eigenvalues.
\end{Assumption}

In practice, 
\textcolor{blue}{
Assumptions \ref{Ass:matrix} and \ref{Ass:matrix:k} hold if that the available
measurements of the system states are sufficiently rich. To be precise, Assumptions \ref{Ass:matrix} and \ref{Ass:matrix:k} hold for the
common special case $S_d = 0$ if i) $\mathcal{Z}$ (Assumption 4) or $X$ (Assumption 5) has full row rank, ii)
$B_w$ is invertible, i.e., $n_w = n$, and iii) the disturbance $W$
generating the data satisfies a strict version of the quadratic
matrix inequality in \eqref{data:disturbance}, i.e., $ W Q_d W^T +R_d \succ 0$.  The proof has been given in \cite[Appendix B]{Berberich2020}.
}
{\color{blue}\begin{Remark}[\emph {Estimation of state derivatives}]
In practice, the time derivatives $\dot{X}$ are typically not available.
However, by means of finite differences and with dense recordings of the system states $X$, a procedure for estimating
the state derivatives has been discussed in \cite{Berberich2020}. Specifically, one can approximate  $\dot{x}(T_i)$ by  $\underline{\dot{x}}(T_i)=\frac{x(T_{i+1})-x(T_{i})}{T_{i+1}-T_{i}}$ with error bound $\|\underline{\dot{x}}(T_i)-\dot{x}(T_i)\|_2 \leq \frac{\bar{a}(T_{i+1}-T_{i})}{2}\Big[\bar{a}\|x(T_i)\|_2+(1+\frac{\bar{a}(T_{i+1}-T_{i})}{3}\bar{b})\|u(T_i)\|_2\Big]$, where $\bar{a}$ and $\bar{b}$ are known bounds for $\|A\|_2\leq \bar{a}$ and $\|B\|_2\leq \bar{b}$. Alternatively, several derivative estimation methods based on continuous-time system identification can be found in \cite{Garnier2003}. In this paper,
we focus on design and analysis of data-driven event-triggered control, where data of the derivative are available, but exploring further possibilities to estimate $\dot{X}$ is an interesting issue for future research.
\end{Remark}}

\begin{Remark}[\emph {Related work}]
Lemma \ref{Lemma:system:data} from \cite{Berberich2020} contributes a purely data-based
representation for sets of continuous-time systems using only some pre-collected input-state data.
Compared to other approaches for data-driven control of continuous-time systems, e.g. in \cite{persis2020}, Lemma \ref{Lemma:system:data} enjoys reduced computational complexity as well as enhanced robustness against additive noise.
Besides, thanks to its form of linear matrix inequalities, the data-based representation in \eqref{data:represent} can be naturally married with  model-based sampled-data control methods, such as the discrete-time approach, input-delay approach, input-output approach, and impulsive systems approach, to study data-driven control of sampled-data systems; see, e.g., \cite{Berberich2020,wildhagen2021datadriven,wildhagen2021improved}.
Overall, the work of \cite{Berberich2020} put forward a unifying
data-driven control framework for continuous-time sampled-data systems with time delays.
In this paper, we build on Lemmas \ref{Lemma:system:data} and \ref{Lemma:system:data:K} to investigate data-driven control of event-triggered sampled-data systems with delays by developing a data-driven looped-functional analysis.
\end{Remark}

\subsection{Dynamic event-triggering transmission scheme}
In the sampled-data control system in Fig. \ref{FIG:structure}, a dynamic event-triggering  module is employed to dictate the transmission instants $\{t_k\}_{k\in \mathbb{N}}$.
In contrast to classical time-triggered transmission schemes, \emph {only} ``informative'' measurements are transmitted by event-triggered schemes to save transmission resources.
Concretely, our proposed dynamic event-triggering scheme consists of two steps. The first step, similar to \emph {time-triggering scheme}, is to sample the system state at time instants $\{t_k+i h-d\}_{i,k \in \mathbb{N}}$, where the sampling interval $h$ satisfies $0< \underline{h} \leq h \leq \bar{h}$.
In the second step, that is \emph {event-triggering scheme},
the sampled signal $x(t_k+i h-d)$ is to be checked for the following condition
{\color{blue}
\begin{align}\label{sys:judgement}
&\eta(\tau_j-d)+\theta \rho(\tau_j-d)<0,
\end{align}
where 
$\rho$ is a discrete-time function $\rho: \{\tau_j-d\}_j\rightarrow \mathbb{R}$, defined as follows
\begin{equation}\label{trigger:func}
\begin{aligned}
\rho:=&~
\sigma_1 x^T(\tau_j-d)\Omega x(\tau_j-d)-e^T(\tau_j-d)\Omega e(\tau_j-d)\\
&+\sigma_2 x^T(t_k-d)\Omega x(t_k-d),
\end{aligned}
\end{equation}
where $\tau_j=t_k+jh$, $j=0,1,\cdots,m_k$, with $m_k=\frac{t_{k+1}-t_k}{h}-1$;
$\Omega\succeq0$ is some weight matrix; $\sigma_1$, $\sigma_2$, and $\theta\geq0 $ are parameters to be designed; $e(\tau_j-d):=x(t_k-d)-x(\tau_j-d)$ denotes the error between sampled signals $x(t_k-d)$ at the latest transmission
instant and $x(\tau_j-d)$ at the current sampling instant; and,
$\eta(\tau_j-d)$ is a dynamically evolving  parameter $\eta: \{\tau_j-d\}_j\rightarrow \mathbb{R}$, satisfying the following difference equation
\begin{equation}\label{sys:dynamic}
\eta(\tau_{j+1}-d)-\eta(\tau_{j}-d)=-\lambda\eta(\tau_j-d)+\rho(\tau_j-d)
\end{equation}}
for which $\eta(0)\geq0$ and $\lambda>0$ are given.

If the condition \eqref{sys:judgement} is met, the current sampled state $x(\tau_j-d)$ is transmitted to the controller, and a zero-order holder (ZOH) follows to maintain it within the
interval $[t_k,t_{k+1})$. Then, the triggering algorithm is updated and detects the next sampled datum.
\textcolor{blue}{In summary, the dynamic event triggering condition can be written as follows
\begin{align}\label{sys:trigger}
t_{k+1}=t_k+h  \min_{j\in \mathbb{N}}\Big\{j>0\Big|\eta(\tau_j-d)+\theta\rho(\tau_j-d)<0\Big\}.
\end{align}
In Fig. \ref{FIG:evolution:event}, a scheme illustrating the periodic sampling with the proposed triggering transmission scheme is provided. The system states are sampled periodically. Then, the delayed sampled data are transmitted to the controller and arrive at $\{t_k\}_k$, if the condition \eqref{sys:judgement} is satisfied, i.e., at time instants $t_1-d$, $t_2-d$, and $t_3-d$.}
\begin{figure}\color{blue}
	\centering
\includegraphics[width=\columnwidth]{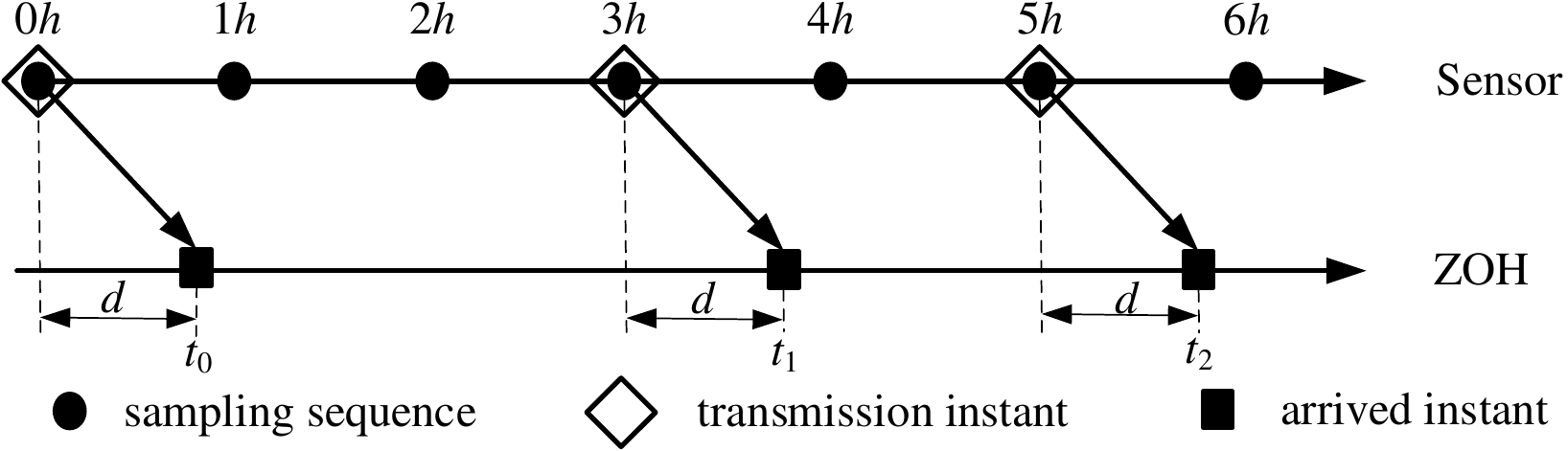}
	\caption{Evolution of sampling and transmission events.}
	\label{FIG:evolution:event}
\end{figure}

The following result states that the value of $\eta(\tau_j-d)$ remains nonnegative for system \eqref{sys:sampling} under the triggering condition \eqref{sys:trigger} with $t\geq0$ if the initial value $\eta(0)\geq0$.

{\color{blue}\begin{Lemma}[\emph {Non-negativity}]\label{lemma:nonneg.dynam}
Let $\eta(0)\ge0$, $\lambda>0$ be positive constants, $\Omega=\Omega^T\succeq 0$ be a positive semi-definite matrix,
and $\theta>0$ be a positive constant satisfying $1-\lambda-\frac{1}{\theta}\geq0$, or $\theta=0$.
Then, it holds that $\eta(\tau_j-d)\geq0$, for all $j\in \mathbb{N}$ under the triggering condition in \eqref{sys:trigger}.
\end{Lemma}

\begin{proof}
According to the event-triggering condition in \eqref{sys:judgement} for $t\geq0$, the condition is not met, which ensures that 
\begin{align}\label{Lemma1:1}
\eta(\tau_j-d)+\theta \rho(\tau_j-d) \geq0, ~\forall j\in \mathbb{N}.
\end{align}
If $\theta=0$, it obviously holds that $\eta(\tau_j-d)\geq0$. Hence, let us focus on $\theta>0$. Inequality \eqref{Lemma1:1} suggests that
\begin{align}\label{Lemma1:2}
\rho(\tau_j-d) \geq -\frac{1}{\theta}\eta(\tau_j-d).
\end{align}
From \eqref{sys:dynamic}, 
it is obtained that
\begin{align}\label{Lemma1:3}
\eta(\tau_{j+1}-d)\geq (1-\lambda)\eta(\tau_j-d)-\frac{1}{\theta}\eta(\tau_j-d).
\end{align}
By induction along with the initial condition $\eta(0)\geq0$ and the assumption $1-\lambda-\frac{1}{\theta}\geq0$, it follows that $\eta(\tau_j-d)\geq0$, for all $j\in \mathbb{N}$,
completing the proof.
\end{proof}}

\begin{Remark}[\emph {Time delay}]
\textcolor{blue}{Time-delays \cite{Lee2020,Sun2010,Zeng2015a,Zeng2019} are frequently introduced into control loops when implementing practical control systems through communication networks (see the structure in Fig. \ref{FIG:structure});
e.g., constant communication delays in the sensor-to-controller channel have been considered in event-triggered systems \cite{ZHANG201555} and sampled-data systems \cite{Zeng2019}.
Time delays can endanger stability properties of the closed loop.
Thus, an important objective of stability analysis is to find an admissible delay region $[\underline{d},\bar{d}]$. 
In this paper, we assume that the upper bound $\bar{d}$ of the constant delay in the sensor-to-controller channel satisfies  $\bar{d}\in [0,h)$.}
Therefore, it is clear that
$t_k<t_{k+1}$, which implies that no
disordered transmission packets occur.
For simplicity, the computation delays in the controller and sensor are neglected.
\end{Remark}
\begin{Remark}[\emph {Comparison of different event-triggering  schemes}]\label{remark:trigger}
{\color{blue}The event-triggering scheme  \eqref{sys:trigger} can be seen as a discrete-time generalization of the continuous dynamic periodic one studied in \cite{Liu2018}. Different from \cite{Liu2018}, the dynamic variable $\eta(t)$ only changes at discrete sampling points $\tau_j$, which eliminates the computational burden caused by continuously computing $\eta(t)$ in \cite{Liu2018}.} 
Besides, our proposed scheme subsumes several existing ones as special cases. For instance,
if $d=0$, $\sigma_1=0$, and $\theta$ goes to $\infty$, the condition \eqref{sys:trigger} becomes the discrete event-triggering scheme in \cite{Yue2013}.
By setting $d=0$, $\sigma_1=0$, and $h=0$, our scheme boils down to the dynamic event-triggering
scheme in \cite{Girard2015}; and, it also degenerates to a periodic transmission scheme when $d=0$, $\sigma_1=0$, $\sigma_2=0$, and $\theta$ goes to $\infty$. Thus, our triggering scheme unifies existing event-triggering schemes, and is expected to save transmission resources.
\end{Remark}

\section{Main Results}\label{sec:mainresults}
In this section, we analyze the stability of system \eqref{sys:sampling} under the transmission scheme in \eqref{sys:trigger}, and derive an upper bound on the sampling interval $h$. A data-based stability condition is then reproduced from the model-based condition, and used for co-designing the triggering matrix $\Omega$ as well as the controller gain $K$ from data.

{\color{blue}\begin{Lemma}\label{lemma:inquality}
For any vectors $\vartheta\in\mathbb{R}^m$, matrices $R=R^T\in\mathbb{R}^{n\times n}\succ0$, $N\in\mathbb{R}^{m\times 2n}$, scalars $\alpha\leq\beta \in \mathbb{N}$, and
a differentiable function $x:[\alpha,\,\beta]\rightarrow \mathbb{R}^{n}$,
the following integral inequalities hold true
\begin{align*}
-\int_{\alpha}^{\beta} \dot x^T(s)R \dot x(s)ds
\leq &(\beta-\alpha)\vartheta^{T}N \mathcal{R}^{-1} N^{T}\vartheta\\
&+{\rm Sym}\left\{\vartheta^{T}N\Pi\right\}
\end{align*}
where 
\begin{align*}
\mathcal{R}&:={\rm diag}\left\{R, \,3R\right\},\\
\Pi&:=\left[x^T(\beta)\!-\!x^T(\alpha), x^T(\beta)\!+\!x^T(\alpha)\!- \!2\int_{\alpha}^{\beta} \frac{x^T(s)}{\beta-\alpha}ds\right]^T.
\end{align*}
\end{Lemma}
Lemma \ref{lemma:inquality} can be cast as a special case of \cite[Lemma 1]{Zeng2015a}, whose proof is omitted here.
Lemma \ref{lemma:inquality} together with the looped-functional approach \cite[Theorem 1]{Seuret2012} are crucial ingredients of our theoretical analysis.}

\subsection{Model-based stability analysis}
We begin by developing a model-based stability analysis of the event-triggered system \eqref{sys:sampling} in this subsection, where matrices $A$ and $B$ are firstly assumed \emph{known}.
The proposed triggering strategy \eqref{sys:trigger} comprises a step of time-trigger with sampling interval $h$. \textcolor{blue}{A bound on the MSI $h$ can be computed through the following stability condition.}
{\color{blue}\begin{Theorem}[\emph {Model-based condition}]\label{Th1}
For given scalars $\bar{h}>\underline{h}>0$, $\bar{d}>\underline{d}>0$, $ \sigma_{1}>0$, $\sigma_{2}>0$, $\lambda>0$, and $\theta>0$ satisfying $1-\lambda-\frac{1}{\theta}\geq0$, or $\theta=0$, system \eqref{sys:sampling} is asymptotically stable under the triggering condition in \eqref{sys:trigger}, and $\eta(\tau_j-d)$ converges to the origin for any $\eta(0)\geq 0$, if there exist matrices $P\succ0$, $Z\succ0$, $T\succ0$, $R_1\succ0$, $R_2\succ0$, $\Omega\succ0$,
$S=S^T$, $N$, $M_1$, $M_2$, $F$, such that the following LMIs hold for all $h\in\{\underline{h}, \bar{h}\}$ and $d\in\{\underline{d}, \bar{d}\}$
\begin{align}{\label {Th1:LMI1}}
&\left[
  \begin{array}{ccc}
    \Xi_0+h\Xi_a+\Psi+\mathcal{O}  & d N \\
    \ast & d \mathcal{T}
  \end{array}
\right]\prec0
\end{align}
\begin{align}{\label {Th1:LMI2}}
&\left[
 \begin{array}{ccc}
    \Xi_0+h\Xi_b+\Psi+\mathcal{O}  & dN & h\mathcal{M}\\
    \ast & d\mathcal{T} &0\\
    \ast & \ast & h\mathcal{R}
  \end{array}
  \right]\prec0
\end{align}
where
\begin{align*}
\Xi_0:=&~{\rm Sym}\Big\{{\Pi}_1^{T}P{\Pi}_3+d{\Pi}_2^{T}P{\Pi}_3+N{\Pi}_6+M_1\Pi_{8}+M_2\Pi_{9}\Big\}\\
&+L_{3}^T\mathcal{T} L_{3}+{\Pi}_{4}^{T}Z{\Pi}_4-{\Pi}_{5}^{T}Z{\Pi}_5,\cr
\Xi_a:=&~\Pi_7^T S \Pi_7+L_{3}^TR_1L_{3}+L_{4}^TR_2L_{4},~
\Xi_b:=-\Pi_7^T S \Pi_7,\\
\Psi:=&~{\rm Sym}\big\{F(AL_1+BL_{10}-L_{3})\big\},\\
\mathcal{O}:=&~\sigma_1 L_7^T \Omega L_7
+\sigma_2 L_{10}^T\Omega L_{10}- (L_7-L_{10})^T \Omega(L_7-L_{10}),\\
\mathcal{R}:=&~{\rm diag}\left\{-R_1, \,-3R_1, \,-R_2, \,-3R_2\right\},\cr
\mathcal{T}:=&~ {\rm diag}\left\{-T, \,-3T\right\},~~~~
\mathcal{M}:=\left[M_1 ~M_2\right],\\
%
{\Pi}_1:=&~[L_1^T,\, L_{2}^T,\,L_{0}^T]^T,~~~~~~\,
{\Pi}_2:=[L_0^T,\, L_{0}^T,\,L_{4}^T]^T,\\
{\Pi}_3:=&~[l_{3}^T,\, L_{4},\,L_{1}^T-L_{2}^T]^T,~
{\Pi}_4:=[L_1^T,\, L_{3}^T],\\
{\Pi}_5:=&~[L_{2}^T,\, L_{4}^T],~~
{\Pi}_6:=[L_1^T-L_{2}^T,\, L_1^T+L_{2}^T-2L_{5}^T]^T,\\
\Pi_7:=&~[L_6^T,\, L_7^T],~~
\Pi_8:=[L_1^T-L_{6}^T,\, L_1^T+L_{6}^T-2L_{8}^T]^T,\\
\Pi_9:=&~[L_2^T-L_{7}^T,\, L_2^T+L_{7}^T-2L_{9}^T]^T,\\
L_i:=&~\left[0_{n\times (i-1)n}, \,I_n, \,0_{n\times (10-i)n} \right], \;i=1, \ldots,10, \\
L_0:=&~0_{n\times10n}.
\end{align*}
\end{Theorem}}

\begin{proof}
Considering the interval $[\tau_j, \tau_{j+1})$ that satisfies $[t_k,\,t_{k+1})=\bigcup \limits_{j=0}^{j=m} [\tau_j, \tau_{j+1})$,
we choose the following functional for system \eqref{sys:sampling}
\begin{align}{\label {Th1:Vt}}
{\color{blue}V(x,t)=V_a(t)+V_l(t),~t\in [\tau_j, \tau_{j+1})}
\end{align}
where
\begin{align}\label{Th1:Va}
V_a(t)=&~\big[x^T(t)~ x^T(t-d) ~\int_{t-d}^tx^T(s)ds\big]P[\cdot]^T \notag\\
&+\int_{t-d}^t \phi^T(s)Z\phi(s)ds+\int_{t-d}^t \int_s^{t}\dot x^T(\upsilon)T {\color{blue}\dot x(\upsilon)}d\upsilon\,ds
\end{align}
{\color{blue} with  $\phi(s):=[x^T(s), \,\dot x^T(s)]^T$, $P\succ0$, $Z\succ0$, and $T\succ0$; and, a looped-functional $V_l(t)$ is constructed as follows
\begin{equation}\label{Th1:W}
\begin{aligned}
V_l(t)=&~(\tau_{j+1}-t)(t-\tau_{j})
\big[x^T(\tau_{j})~x^T(\tau_{j}-d)\big]S[\cdot]^T\\
&+(\tau_{j+1}-t)\int_{\tau_{j}}^{t} \dot x^T(s)R_1\dot x(s)ds\\
&+(\tau_{j+1}-t)\int_{\tau_{j}-d}^{t-d} \dot x^T(s)R_2\dot x(s)ds\\
\end{aligned}
\end{equation}
where $S=S^T$, $R_1\succ0$, and $R_2\succ0$ are to be designed.
The functional $V_a(t)$ is positive definite, while $V_l(t)$ satisfies $ \lim \limits_{t\rightarrow \tau_{j}}V_l(t)=\lim \limits_{t\rightarrow \tau_{j+1}}V_l(t)$. Following the loop-functional approach \cite{Seuret2012}, we have the next steps.}

Taking the derivative of $V(t)$ along the trajectory of system \eqref{sys:sampling} yields
\begin{align}{\label {Th1:Vd}}
\dot V(t)=\frac{d}{dt} \big[V_a(t)+ V_l(t)\big],
\end{align}
where $\dot \eta(t)$ is given in \eqref{sys:dynamic}, and
{\color{blue}\begin{align*}
\dot V_a(t)=\xi^T(t)\Big[&2({\Pi}_1^{T}P{\Pi}_3+d{\Pi}_2^{T}P{\Pi}_3)+L_{3}^TT L_{3}+ {\Pi}_{4}^{T}Z{\Pi}_4\\
&-{\Pi}_{5}^{T}Z{\Pi}_5\Big] \xi(t)-\int_{t-d}^t\dot x^T(s)T\dot x(s)ds,\\
\dot V_{l}(t)=\xi^T(t)\Big[&(\tau_{j+1}-t)\big(\Pi_7^T S \Pi_7+L_{3}^TR_1L_{3}+L_{4}^TR_2L_{4}\big)\\
&-(t-\tau_{j})\Pi_7^T S \Pi_7\Big]\xi(t)-\!\int_{\tau_{j}}^{t} \dot x^T(s)R_1\dot x(s)ds\\
&-\int_{\tau_{j}-d}^{t-d} \dot x^T(s)R_2 \dot x(s)ds.
\end{align*}}
where the notation $\xi(t)$ is given as
\begin{align*}
&\xi(t):=\bigg[x^T(t), x^T(t-d), \dot x^T(t), \dot x^T(t-d), \int_{t-d}^{t}\frac{x(s)}{d}ds,\\ &x^T(\tau_{j}),\!
x^T(\tau_{j}\!-\!d), \! \int_{t}^{\tau_{j}}\!\frac{x^T(s)}{t\!-\!\tau_{j}}ds,\!\int_{t\!-\!d}^{\tau_{j}\!-\!d}\!\frac{x^T(s)}{t-\tau_{j}}ds, \! x^T(t_k\!-\!d)\!\bigg]^T
\end{align*}

Note that there are some integral terms present in $\dot V(t)$.
According to Lemma \ref{lemma:inquality}, those terms satisfy that
\begin{align}
&-\int_{t-d}^t\dot x^T(s)T\dot x(s)ds\leq
\xi^T(t)\Big(d N {\mathcal{T}}^{-1} N^T + 2N{\Pi}_6 \Big)\xi(t)\\
&-\int_{\tau_{j}}^{t} \dot x^T(s)R_1 \dot x(s)ds-\int_{\tau_{j}-d}^{t-d} \dot x^T(s)R_1 \dot x(s)ds \notag \\
&~\leq \xi^T(t)\Big((t-\tau_{j}) \mathcal{M} \mathcal{R}^{-1} \mathcal{M}^T + 2M_1\Pi_{8}+2M_2\Pi_{9} \Big)\xi(t)
\end{align}
where $N$, $M_1$ and $M_2$ are suitable matrices.

According to \eqref{sys:sampling}, it is clear that for $t\in [\tau_j, \tau_{j+1})$
\begin{align}{\label {Th1:zero}}
0&=2\xi^T(t)F \big[A x(t)+BKx(t_k-d)-\dot x(t) \big]\notag\\
&=2\xi^T(t)F\big(Al_1+BKL_{10}-L_{3}\big)\xi(t)
\end{align}
where $F$ is a matrix of dimension $10n\times n$.

{\color{blue}
Summing up \eqref{Th1:Vd}---\eqref{Th1:zero} gives rise to
\begin{equation}{\label {Th1:sum}}
\begin{aligned}
\dot V(t)\leq \xi^{T}(t)\Big[&(\tau_{j+1}-t)\Xi_{a}+(t-\tau_j)(\Xi_{b}+\mathcal{M} \mathcal{R}^{-1} \mathcal{M}^T)\\
&+\Xi_{0}+\Psi+d N {\mathcal{T}}^{-1} N^T\Big]\xi(t).
\end{aligned}
\end{equation}

In light of the triggering condition \eqref{sys:trigger} for $t\in[t_k,\,t_{k+1})$, 
Lemma \ref{lemma:nonneg.dynam} asserts that $\eta(\tau_j-d)\geq0$ for $\eta(0)\ge0$, $\lambda>0$, and $\theta>0$ satisfying $1-\lambda-\frac{1}{\theta}\geq0$, or $\theta=0$.
 Besides, according to the equation \eqref{sys:dynamic}, it holds that
\begin{equation*}
 \begin{aligned}
\eta(\tau_{j+1}-d)-\eta(\tau_j-d)=-\lambda\eta(\tau_j-d)+\rho(\tau_j-d).
\end{aligned}
\end{equation*}
From \eqref{trigger:func}, $\rho(\tau_j-d)$ is written to be $\rho(\tau_j-d)=\xi^{T}(t) \mathcal{O}\xi(t)$. Thus, it easy to be derived that
 \begin{align}\label{Th1:triggering}
\eta(\tau_{j+1}-d)-\eta(\tau_j-d)\leq \xi^{T}(t)\mathcal{O}\xi(t).
\end{align}

Combining \eqref{Th1:sum} with \eqref{Th1:triggering}, we have that
\begin{equation}\label{Th1:Vd:sum}
 \begin{aligned}
\dot V(t)+\eta(&\tau_{j+1}-d)-\eta(\tau_j-d)\leq\\
& \xi^{T}(t)\left[\frac{\tau_{j+1}-t}{h}\Upsilon_1(h)
+\frac{t-\tau_j}{h}\Upsilon_2(h) \right]\xi(t)
\end{aligned}
\end{equation}
where
\begin{align*}
\Upsilon_1(h)&=\Xi_{0}+\Psi+\mathcal{O}+h\Xi_{a}+d N \mathcal{T}^{-1} N^T\\
\Upsilon_2(h)&=\Xi_{0}+\Psi+\mathcal{O}+h\Xi_{b} +h \mathcal{M} \mathcal{R}^{-1} \mathcal{M}^T+d N \mathcal{T}^{-1} N^T.
\end{align*}

Using the Schur Complement Lemma, it can be further deduced that $\Upsilon_1(h)\prec0$ and $\Upsilon_2(h)\prec0$ are equivalent
to \eqref{Th1:LMI1} and \eqref{Th1:LMI2}, which are affine in $h$ and $d$. Thus, \eqref{Th1:LMI1} and \eqref{Th1:LMI2} at the vertices of $(h,d)\in[\underline{h},\bar{h}]\times[\underline{d},\bar{d}]$ ensure $\dot V(t)+ \eta(\tau_{j+1}-d)-\eta(\tau_j-d)<0$ for all $(h,d)\in[\underline{h},\bar{h}]\times[\underline{d},\bar{d}]$.
It follows from \cite[Theorem 1]{SEURET2012177} that
\begin{equation}{\label {Th1:vj}}
\begin{aligned}
\int_{\tau_j}^{\tau_{j+1}} &[\dot V(s)+\eta(\tau_{j+1}-d)-\eta(\tau_j-d)]ds=\\
&V_a(\tau_{j+1})-V_a(\tau_{j})+h[\eta(\tau_{j+1}-d)-\eta(\tau_j-d)]<0
\end{aligned}
\end{equation}
which implies {\label {Th1:vj:equal}}
\begin{equation}
V_a(\tau_{j+1})+h\eta(\tau_{j+1}-d)<V_a(\tau_{j})+h\eta(\tau_j-d), \forall j\in \mathbb{N}.
\end{equation}
Finally, using $V_a(\tau_j)>0$ and $\eta(\tau_j-d)>0$, we can conclude that system \eqref{sys:sampling} and $\eta(\tau_j-d)$ converge to the origin under our
transmission scheme.
This ends the proof.}
\end{proof}

\textcolor{blue}{Theorem \ref{Th1} is derived based on  a novel looped-functional, using which larger bounds on  $h$ can be obtained, compared to the stability condition using only a common Lyapunov functional. Note  that Theorem \ref{Th1} guarantees stability for arbitrary values  of $h\in[\underline{h},\bar{h}]$ and $d\in[\underline{d},\bar{d}]$.
This is due to the fact that \eqref{Th1:LMI1} and \eqref{Th1:LMI2} in Theorem \ref{Th1} are affine in $h$ and $d$, and we verify the conditions  simultaneously for $h=\underline{h}$, $h=\bar{h}$, $d=\underline{d}$, and $d=\bar{d}$.}

\textcolor{blue}{In networked control systems, the values of the constant sampling interval $h$ and time-delay $d$ may be changed in a set  $[\underline{h},\bar{h}]\times[\underline{d},\bar{d}]$ due to undesired network effects. Thus,
admissible intervals of the time-delay and the sampling interval are required in co-designing the controller and the triggering scheme that guarantee the desired system performance. Theorem \ref{Th1} provides a systematic approach for computing such regions, based directly on measured data. 
It is also an interesting topic to increase the allowable ranges of $h$ and $d$, which is related to the conservatism of the stability conditions. This issue is discussed in the following remark.}

{\color{blue}\begin{Remark}[\emph {Discussion of looped-functional}]\label{R:loop}
The looped-functional \cite{Seuret2012} approach has been widely used for stability analysis of sampled-data control systems. Compared to the continuous-time Lyapunov stability theorem, the looped-functional approach provides less conservative stability conditions. Recently, various  types of looped-functionals have been proposed, such as two-sided looped-functional \cite{Zeng2017}, integral-based looped-functional \cite{Lee20177}, extended looped-functional \cite{Park2020}, and general looped-functional \cite{Wang2021}.
In the proof of Theorem \ref{Th1}, a simple type of one-sided delay-dependent looped-functional $V_l(t)$ [cf. \eqref{Th1:W}] is constructed.
It is easy to prove that $ \lim \limits_{t\rightarrow \tau_{j}}V_l(t)=\lim \limits_{t\rightarrow \tau_{j+1}}V_l(t)$, which asserts that $V_l(t)$ is a looped-functional for $t\in[\tau_j, \tau_{j+1})$. By Theorem \ref{Th1},  upper bounds on $h$ for varying transmission delays $d$ can be readily computed. The practical applicability of Theorem \ref{Th1} is illustrated with numerical examples in Section \ref{sec:example}.
It should be mentioned that larger allowable bounds on $h$ can be obtained by using more complicated looped-functional methods.
Besides, we only focus on closed-loop stability in Theorem \ref{Th1}, but it is straightforward to derive other  performance guarantees, including e.g., on the closed-loop $\mathcal{L}_2$-gain, using similar arguments as in \cite{Wang2021mixed}.
\end{Remark}}

\subsection{Data-based stability analysis}
We now derive a data-based stability certificate for the event-triggered control system \eqref{sys:sampling} with \emph{unknown} system matrices $A$ and $B$ (cf. Assumption \ref{Ass:matrix:AB}).
The main idea is to rebuild a system expression using the data $\{\dot x(T_i),\,x(T_i),\,u(T_i)\}^{\rho}_{i=1}$ to replace the matrix $(A~B)$-based representation in \eqref{sys:LTI}.
According to this,
a data-based system representation given in Lemma \ref{Lemma:system:data}, combined with the  model-based stability condition in Theorem \ref{Th1}, is employed to obtain a data-based stability condition.
\begin{Theorem}[\emph {Data-based condition with unknown $A$ and $B$}]\label{Th2}
For given scalars $\bar{h}>\underline{h}>0$, $\bar{d}>\underline{d}>0$, $\theta\geq0$, $\sigma_{1}>0$, $\sigma_{2}>0$, {\color{blue}$\lambda>0$, and $\theta>0$ satisfying $1-\lambda-\frac{1}{\theta}\geq0$, or $\theta=0$, system \eqref{sys:sampling} is asymptotically stable under the triggering condition in \eqref{sys:trigger} for any $[A ~B]\in \Sigma_s$, and $\eta(\tau_j-d)$ converges to the origin for any $\eta(0)\geq 0$}, if there exist a scalar $\varepsilon>0$, and matrices $P\succ0$, $Z\succ0$, $T\succ0$, {\color{blue}$R_1\succ0$, $R_2\succ0$, $\Omega\succ0$,
$S=S^T$, $N$, $M_1$, $M_2$,} $F$, such that the following LMIs hold for all $h\in\{\underline{h}, \bar{h}\}$ and $d\in\{\underline{d}, \bar{d}\}$
{\color{blue}
\begin{align}{\label {Th2:LMI1}}
&\left[
  \begin{array}{cccc}
    \mathcal{G}_3& \mathcal{G}_2+F^T& 0\\
    \ast & \mathcal{G}_1+\Xi_{0}+h\Xi_a+\tilde{\Psi}+\mathcal{O}  & dN \\
    \ast & \ast & d\mathcal{T} \\
  \end{array}
\right]\prec0
\end{align}
\begin{align}{\label {Th2:LMI2}}
&\left[\!\!
 \begin{array}{cccc}
     \mathcal{G}_3& \mathcal{G}_2+F^T& 0& 0\\
    \ast & \mathcal{G}_1+\Xi_0+h\Xi_b+\tilde{\Psi} +\mathcal{O}  & dN & h\mathcal{M}\\
    \ast & \ast & d \mathcal{T} &0\\
    \ast & \ast & \ast & h \mathcal{R}
  \end{array}
  \!\!\right]\prec0
\end{align}}
where
\begin{align*}
\tilde{\Psi}&:={\rm Sym}\big\{-FL_{3}\big\}\\
\mathcal{G}_1&:=\varepsilon\mathcal{Y}_1^T\tilde{\Theta}\mathcal{Y}_1,~
\mathcal{G}_2:=\varepsilon\mathcal{Y}_2^T\tilde{\Theta}\mathcal{Y}_1,~
\mathcal{G}_3:=\varepsilon\mathcal{Y}_2^T\tilde{\Theta}\mathcal{Y}_2\\
\tilde{\Theta}&:=\left[\begin{array}{cc}-\tilde{R}_c & \tilde{S}_c^T\\\ast &  -\tilde{Q}_c \\\end{array}\right],~
\left[\begin{array}{cc}\tilde{Q}_c & \tilde{S}_c\\\ast & \tilde{R}_c \\\end{array}\right]:=
\left[\begin{array}{cc}Q_c & S_c\\\ast & R_c \\\end{array}\right]^{-1}\\
\mathcal{Y}_1&:=[0~L_1^T~ (KL_{10})^T]^T,~
\mathcal{Y}_2:=[I~0~ 0 ]^T.
\end{align*}
\end{Theorem}

\begin{proof}
According to \eqref{Th1:Vd} in the proof of Theorem \ref{Th1}, we have a conclusion that the system \eqref{sys:sampling} under the event-triggering scheme \eqref{sys:trigger} is stable if $\Upsilon_1(h)\prec0$ and $\Upsilon_2(h)\prec0$. By matrix decomposition, $\Upsilon_1(h)$ and $\Upsilon_2(h)$ are rewritten as
\begin{align}\label{Th2:decomposition1}
\Upsilon_1(h)&=\left[\begin{array}{cc}Al_1+BKl_{19}\\I \\\end{array}\right]^T
\left[\begin{array}{cc}0 & F^T\\\ast & \Upsilon_1(h)+\tilde{\Psi}- \Psi\\\end{array}\right][\cdot]\\
\label{Th2:decomposition2}
\Upsilon_2(h)&=\left[\begin{array}{cc}Al_1+BKl_{19}\\I \\\end{array}\right]^T
\left[\begin{array}{cc}0 & F^T\\\ast & \Upsilon_2(h)+\tilde{\Psi}- \Psi\\\end{array}\right][\cdot]
\end{align}

Besides, applying the dualization lemma \cite[Lemma 4.9]{Scherer2000} to the system representation in \eqref{data:represent} with Assumption \ref{Ass:matrix}, it can be proved that $[A~B]\in\Sigma_s$ if and only if
\begin{align}\label{Th2:dualization}
\left[\!\!\begin{array}{cc}[A~B]\\I \\\end{array}\!\!\right]^T
\left[\begin{array}{cc}-\tilde{R}_c & \tilde{S}_c^T\\\ast &  -\tilde{Q}_c \\\end{array}\right]
\left[\!\!\begin{array}{cc}[A~B]\\I \\\end{array}\!\!\right]
\succeq0.
\end{align}


Using the full-block S-procedure \cite{Sche2001}, we have $\Upsilon_1(h)\prec0$ and $\Upsilon_2(h)\prec0$ for any $[A ~B]\in\Sigma_s$ if there exists a scalar $\varepsilon>0$ such that
 \begin{align}\label{Th2:fullblock1}
&\tilde{\Upsilon}_1(h)+
\varepsilon \left[\begin{array}{cc}\mathcal{Y}_2^T\tilde{\Theta}\mathcal{Y}_2 & \mathcal{Y}_2^T\tilde{\Theta}\mathcal{Y}_1\\\ast &  \mathcal{Y}_1^T\tilde{\Theta}\mathcal{Y}_1 \\\end{array}\right]\prec0 \\ \label{Th2:fullblock2}
&\tilde{\Upsilon}_2(h)+
\varepsilon \left[\begin{array}{cc}\mathcal{Y}_2^T\tilde{\Theta}\mathcal{Y}_2 & \mathcal{Y}_2^T\tilde{\Theta}\mathcal{Y}_1\\\ast &  \mathcal{Y}_1^T\tilde{\Theta}\mathcal{Y}_1 \\\end{array}\right]\prec0.
\end{align}

Then, the Schur Complement Lemma renders \eqref{Th2:fullblock1} and \eqref{Th2:fullblock2} equivalent to \eqref{Th2:LMI1} and \eqref{Th2:LMI2}.
Similar to Theorem \ref{Th1}, we conclude that
system \eqref{sys:sampling} is asymptotically stable under triggering condition \eqref{sys:trigger} for any $[A ~B]\in \Sigma_s$, and $\eta(\tau_j-d)$ converges to the origin,
which completes the proof.
\end{proof}

Theorem \ref{Th2} allows us to analyze stability properties of dynamic event-triggered control with delays, without any model knowledge (cf. Assumption \ref{Ass:matrix:AB}).
A possibly large MSI and a triggering matrix $\Omega$ for the triggering condition \eqref{sys:trigger} can be searched for with a given controller gain $K$ by using Theorem \ref{Th2}. 
The application of Theorem \ref{Th2} is simple, requiring only the solution of LMIs which can be constructed based on noisy data. 
For event-triggered systems,
a larger MSI leads to a smaller transmission frequency, which saves network transmission resources.
To this end, we further investigate
a data-based stability condition for system \eqref{sys:sampling} with \emph {unknown} state matrix $A$ and \emph {known} input matrix $B$ (cf. Assumption \ref{Ass:matrix:A}), since any additional prior knowledge may lead to a less conservative data-based stability condition (i.e., a larger MSI) if compared to the one based only on the available data (cf. Theorem \ref{Th2}). The following stability result is established based on Lemma \ref{Lemma:system:data:K} and Assumption \ref{Ass:matrix:k}.




\begin{Theorem}[\emph {Data-based condition with unknown $A$ and known $B$}]\label{Th3}
For given scalars $\bar{h}>\underline{h}>0$, $\bar{d}>\underline{d}>0$, $\sigma_{1}>0$, $\sigma_{2}>0$, {\color{blue}$\lambda>0$, and $\theta>0$ satisfying $1-\lambda-\frac{1}{\theta}\geq0$, or $\theta=0$, system \eqref{sys:sampling} is asymptotically stable under the triggering condition in \eqref{sys:trigger} for any $[A]\in \bar{\Sigma}_s$, and $\eta(\tau_j-d)$ converges to the origin for any $\eta(0)\geq 0$}, if there exist a scalar $\varepsilon>0$, and matrices $P\succ0$, $Z\succ0$, {\color{blue}$R_1\succ0$, $R_2\succ0$, $\Omega\succ0$,
$S=S^T$, $N$, $M_1$, $M_2$,} $F$, such that the following LMIs hold for all $h\in\{\underline{h}, \bar{h}\}$ and $d\in\{\underline{d}, \bar{d}\}$
{\color{blue}\begin{align}{\label {Th3:LMI1}}
&\left[\!\!
  \begin{array}{cccc}
    \hat{\mathcal{G}}_3& \hat{\mathcal{G}}_2+F^T& 0\\
    \ast & \hat{\mathcal{G}}_1+\Xi_{0}+h\Xi_a+\hat{\Psi}+\mathcal{O}  & dN \\
    \ast & \ast & d\mathcal{T}
  \end{array}
\!\!\right]\!\!\prec0
\end{align}
\begin{align}{\label {Th3:LMI2}}
&\left[\!\!
 \begin{array}{cccc}
     \hat{\mathcal{G}}_3& \hat{\mathcal{G}}_2+F^T& \!0& \!0\\
    \ast & \hat{\mathcal{G}}_1+\Xi_0+h\Xi_b+\hat{\Psi}+\mathcal{O} & \!dN & \!h\mathcal{M}\\
    \ast & \ast & \!d\mathcal{T} &\!0\\
    \ast & \ast & \!\ast & \!h\mathcal{R}
  \end{array}
  \!\!\right]\prec0
\end{align}}
where
\begin{align*}
\hat{\Psi}&:={\rm Sym}\big\{F(BKL_{10}-L_{3})\big\}\\
\hat{\mathcal{G}}_1&:=\varepsilon\hat{\mathcal{Y}}_1^T\hat{\Theta}\hat{\mathcal{Y}}_1,~
\hat{\mathcal{G}}_2:=\varepsilon\hat{\mathcal{Y}}_2^T\hat{\Theta}\hat{\mathcal{Y}}_1,~
\hat{\mathcal{G}}_3:=\varepsilon\hat{\mathcal{Y}}_2^T\hat{\Theta}\hat{\mathcal{Y}}_2\\
\hat{\Theta}&:=\left[\begin{array}{cc}-\hat{R}_c & \hat{S}_c^T\\\ast &  -\hat{Q}_c \\\end{array}\right],~
\left[\begin{array}{cc}\hat{Q}_c & \hat{S}_c\\\ast & \hat{R}_c \\\end{array}\right]:=
\left[\begin{array}{cc}\bar{Q}_c & \bar{S}_c\\\ast & \bar{R}_c \\\end{array}\right]^{-1}\\
\hat{\mathcal{Y}}_1&:=[0~L_1^T]^T,~
\hat{\mathcal{Y}}_2:=[I~0]^T.
\end{align*}
\end{Theorem}

\begin{proof}The full proof is similar to that of Theorem \ref{Th2}. Terms $\Upsilon_1(h)$ and $\Upsilon_2(h)$ in \eqref{Th2:decomposition1} and \eqref{Th2:decomposition2} can be rebuilt as follows
\begin{align}\label{Th3:decomposition1}
\Upsilon_1(h)&=\left[\begin{array}{cc}Al_1\\I \\\end{array}\right]^T
\left[\begin{array}{cc}0 & F^T\\\ast & \Upsilon_1(h)+\hat{\Psi}- \Psi \\\end{array}\right]\left[\begin{array}{cc}Al_1\\I \\\end{array}\right]\\
\label{Th3:decomposition2}
\Upsilon_2(h)&=\left[\begin{array}{cc}Al_1\\I \\\end{array}\right]^T
\left[\begin{array}{cc}0 & F^T\\\ast & \Upsilon_2(h)+\hat{\Psi}- \Psi\\\end{array}\right]\left[\begin{array}{cc}Al_1\\I \\\end{array}\right]
\end{align}

 Employing the dualization lemma to the system representation in \eqref{data:represent:K} and Assumption \ref{Ass:matrix:k}, we have that $[A]\in\bar{\Sigma}_s$ if and only if
 \begin{align}\label{Th3:dualization}
&\left[\begin{array}{cc}A\\I \\\end{array}\right]^T
\left[\begin{array}{cc}-\hat{R}_c & \hat{S}_c^T\\\ast &  -\hat{Q}_c \\\end{array}\right]
\left[\begin{array}{cc}A\\I \\\end{array}\right]
\succeq0.
\end{align}

According to the full-block S-procedure, it holds that $\Upsilon_1(h)\prec0$ and $\Upsilon_2(h)\prec0$ for
any $[A]\in \bar{\Sigma}_s$ if there exists a scalar
 $\varepsilon>0$ such that
\begin{align}\label{Th3:fullblock1}
&\left[\begin{array}{cc}0 & F^T\\\ast & \Upsilon_1(h)+\hat{\Psi}- \Psi \\\end{array}\right]+
\varepsilon \left[\begin{array}{cc}\hat{\mathcal{Y}}_2^T\hat{\Theta}\hat{\mathcal{Y}}_2 & \hat{\mathcal{Y}}_2^T\hat{\Theta}\hat{\mathcal{Y}}_1\\\ast &  \hat{\mathcal{Y}}_1^T\hat{\Theta}\hat{\mathcal{Y}}_1 \\\end{array}\right]\prec0 \\ \label{Th3:fullblock2}
&\left[\begin{array}{cc}0 & F^T\\\ast & \Upsilon_2(h)+\hat{\Psi}- \Psi\\\end{array}\right]+
\varepsilon \left[\begin{array}{cc}\hat{\mathcal{Y}}_2^T\hat{\Theta}\hat{\mathcal{Y}}_2 & \hat{\mathcal{Y}}_2^T\hat{\Theta}\hat{\mathcal{Y}}_1\\\ast &  \hat{\mathcal{Y}}_1^T\hat{\Theta}\hat{\mathcal{Y}}_1 \\\end{array}\right]\prec0.
\end{align}

Finally, similar to Theorem \ref{Th1}, we conclude that \eqref{Th3:LMI1} and \eqref{Th3:LMI2} are sufficient stability conditions for system \eqref{sys:sampling} under the triggering condition \eqref{sys:trigger} for any $[A]\in \bar{\Sigma}_s$, and $\eta(\tau_j-d)$ converges to the origin.
\end{proof}

Theorem \ref{Th3} provides a stability condition of system \eqref{sys:LTI} using the prior knowledge of the input matrix $B$ and data. 
Compared to Theorem \ref{Th2} based only on the available data, a larger MSI is able to be obtained by Theorem \ref{Th3} for a given controller gain; see numerical comparisons between Theorems \ref{Th2} and \ref{Th3} for different levels of noise given in Table \ref{Tab:noise} of Section \ref{sec:example}.

\subsection{Data-based co-designing controller and triggering matrix}\label{subsection:controller}
This section provides co-design methods of the controller gain $K$ and the triggering matrix $\Omega$ for system \eqref{sys:sampling} under Assumptions \ref{Ass:matrix:AB} and \ref{Ass:matrix:A}. For Assumptions \ref{Ass:matrix:AB}, 
we want to compute matrices $K$ and $\Omega$ using Theorem \ref{Th2} with a given sampling interval $h$ and the delay $d$.
However, since the matrix $\tilde{Q}_c$ is generally not negative definite,
inequalities \eqref{Th2:LMI1} and \eqref{Th2:LMI2} are actually not convex in $K$.
It is hence hard to design $K$ directly through Theorem \ref{Th2}. Inspired by the iterative approach in \cite[Theorems 4-5]{Berberich2020}, we give a reformulation of Theorem \ref{Th2} that is convex in $K$ for a fixed variable $F$. Therefore, the controller gain can be searched while simultaneously obtaining an MSI under Assumption \ref{Ass:matrix:AB}.

\begin{Theorem}[\emph {Co-designing under unknown $A$, $B$, and fixed matrix $F$}]\label{Th:Co-designing:AB:F}
For any given scalars $\bar{h}>\underline{h}>0$, $\bar{d}>\underline{d}>0$, $\sigma_{1}>0$, $\sigma_{2}>0$, {\color{blue}$\lambda>0$, and $\theta>0$ satisfying $1-\lambda-\frac{1}{\theta}\geq0$, or $\theta=0$,
there exists a controller gain $K$ such that
 system \eqref{sys:sampling} is asymptotically stable under the triggering condition in \eqref{sys:trigger} for any pair $[A ~B]\in \Sigma_s$, and that $\eta(\tau_j-d)$ converges to the origin for any $\eta(0)\geq 0$}, provided that there are $\varepsilon>0$, and matrices $P\succ0$, $Z\succ0$, $T\succ0$, {\color{blue}$R_1\succ0$, $R_2\succ0$, $\Omega\succ0$,
$S=S^T$, $N$, $M_1$, $M_2$,} $F$ such that the following LMIs hold for all $h\in\{\underline{h}, \bar{h}\}$ and $d\in\{\underline{d}, \bar{d}\}$
{\color{blue}\begin{align}{\label {Th:controller:ABF:LMI1}}
&\left[\!\!
  \begin{array}{cccc}
    \mathcal{S}_1& \mathcal{S}_2+\big[L_1^T~(KL_{10})^T\big]^T& 0\\
    \ast & \mathcal{S}_3+\Xi_0+h\Xi_a+\tilde{\Psi}+\mathcal{O}  & d N \\
    \ast & \ast & d \mathcal{T} \\
    \ast & \ast & \ast
  \end{array}
\!\!\right]\!\!\prec0
\end{align}
\begin{align}{\label {Th:controller:ABF:LMI2}}
&\left[\!\!
 \begin{array}{cccc}
     \mathcal{S}_1& \mathcal{S}_2+\big[L_1^T~(KL_{10})^T\big]^T& 0& 0\\
    \ast & \mathcal{S}_3+\Xi_{0}+h\Xi_b+\tilde{\Psi}+\mathcal{O}  & d N & h\mathcal{M}\\
    \ast & \ast &d\mathcal{T} &0\\
    \ast & \ast & \ast & h\mathcal{R}
  \end{array}
  \!\!\right]\prec0
\end{align}}
where
\begin{align*}
\mathcal{S}_1&:=\varepsilon\mathcal{D}_1\Theta_s\mathcal{D}_1^T,
\mathcal{S}_2:=\varepsilon\mathcal{D}_1\Theta_s\mathcal{D}_2^T,
\mathcal{S}_3:=\varepsilon\mathcal{D}_2\Theta_s\mathcal{D}_2^T\\
\mathcal{D}_1&:=
\left[\begin{array}{ccc}I & 0& 0\\
0 & 1 & 0\\\end{array}\right],~
\mathcal{D}_2:=
\left[\begin{array}{ccc}0 & 0& F^T\\\end{array}\right].
\end{align*}
\end{Theorem}

\begin{proof}
Similar to the proof of Theorem \ref{Th2}, we reformulate $\Upsilon_1(h)$ and $\Upsilon_2(h)$ as follows
\begin{align}\label{Th:ABF:decomposition1}
&{\Upsilon}_1(h)=
\left[\begin{array}{cc}[FA~FB]^T\\I \\\end{array}\right]^T
\left[\begin{array}{cc}0 & \big[L_1^T~(KL_{10})^T\big]^T\\\ast & \Upsilon_1(h)+\tilde{\Psi}- \Psi \\\end{array}\right]
[\cdot]\\
\label{Th4:ABF:decomposition2}
&{\Upsilon}_2(h)=
\left[\begin{array}{cc}[FA~FB]^T\\I \\\end{array}\right]^T
\left[\begin{array}{cc}0 & \big[L_1^T~(KL_{10})^T\big]^T\\\ast & \Upsilon_2(h)+\tilde{\Psi}- \Psi \\\end{array}\right]
[\cdot].
\end{align}

From the data-based representation \eqref{data:represent}, it holds that
\begin{align}
\left[\!\begin{array}{cc}[A~B]^T\\I \\\end{array}\!\right]^T
  \!\Theta_s\!
  \left[\!\begin{array}{cc}[A~B]^T \\
   I\\\end{array}\!\right]\succeq0.
\end{align}

By the full-block S-procedure, we have ${\Upsilon}_1(h)\prec0$ and ${\Upsilon}_2(h)\prec0$ for any $[A ~B]\in\Sigma_s$ if there exists a scalar $\varepsilon>0$ such that
 \begin{align}\label{Th:ABF:fullblock1}
&\left[\begin{array}{cc}0 & \big[L_1^T~(KL_{10})^T\big]^T\\\ast & \Upsilon_1(h)+\tilde{\Psi}- \Psi \\\end{array}\right]+
\varepsilon \left[\begin{array}{cc}\mathcal{D}_1\Theta_s\mathcal{D}_1^T & \mathcal{D}_1\Theta_s\mathcal{D}_2^T\\\ast &  \mathcal{D}_2\Theta_s\mathcal{D}_2^T \\\end{array}\right]  \notag\\ &\prec0  \\
\label{Th:ABF:fullblock2}
&\left[\begin{array}{cc}0 & \big[L_1^T~(KL_{10})^T\big]^T\\\ast & \Upsilon_2(h)+\tilde{\Psi}- \Psi \\\end{array}\right]+
\varepsilon \left[\begin{array}{cc}\mathcal{D}_1\Theta_s\mathcal{D}_1^T & \mathcal{D}_1\Theta_s\mathcal{D}_2^T\\\ast &  \mathcal{D}_2\Theta_s\mathcal{D}_2^T \\\end{array}\right]\notag \\
&\prec0.
\end{align}

Finally, similar to the proof of Theorem \ref{Th2}, we have a conclusion that \eqref{Th:controller:ABF:LMI1} and \eqref{Th:controller:ABF:LMI2} are sufficient stability conditions for system \eqref{sys:sampling}
under the triggering condition \eqref{sys:trigger} for any $[A ~B]\in \Sigma_s$, and $\eta(\tau_j-d)$ converges to the origin, which completes the proof.
\end{proof}
\begin{Remark}[\emph {Iterative approach}]
\textcolor{blue}{Notice that the inequalities \eqref{Th:controller:ABF:LMI1} and \eqref{Th:controller:ABF:LMI2} in Theorem \ref{Th:Co-designing:AB:F} are quadratic in the matrix $F$.
Hence, we cannot efficiently search for the matrix $K$ by Theorem \ref{Th:Co-designing:AB:F}.
Fixing the matrix $F$ enables the design of the controller; however, the conservatism of the obtained MSI estimate may be significantly increased.
To this end, we propose an iterative approach which alternates 
between solving the LMIs
in Theorems \ref{Th2} and \ref{Th:Co-designing:AB:F} while fixing the matrix $K$ or $F$, respectively.
The effectiveness of this method is illustrated via a numerical example in Section \ref{sec:example},
and the results of MSI bounds in Table \ref{Tab:iteration} show that such an iterative method tremendously improves the MSI bounds
in comparison to using a fixed controller gain as in Theorems \ref{Th2} and \ref{Th3} (see Table \ref{Tab:noise}).}
On the other hand, for the case of having a known input matrix $B$, we can search for a controller gain $K$ with possibly large MSI by alternating between solving the LMIs in Theorem \ref{Th3} while fixing matrix $K$ or matrix $F$, in both of which the LMIs are convex and can be efficiently solved.
Compared with the iterative method between Theorems \ref{Th2} and \ref{Th:Co-designing:AB:F} that are purely based on data, the one using Theorem \ref{Th3} that contains the prior knowledge of matrix $B$ may provide a larger MSI while iteratively searching for the corresponding $K$; see numerical comparisons in Table \ref{Tab:iteration}.
\end{Remark}
In the following, we provide a stability condition that is convex in all decision variables (different from Theorems \ref{Th3}, \ref{Th:Co-designing:AB:F} and \cite[Theorem 5]{Berberich2020}) for co-designing matrices $K$ and $\Omega$, at the price of additional conservatism with respect to the matrix $F$. 
For this purpose, we begin with an algebraically equivalent transformation of system \eqref{sys:sampling}.

Let $G\in \mathbb{R}^{n \times n}$ be a nonsingular matrix, and define $x(t)=Gz(t)$. system \eqref{sys:sampling} is transformed into
\begin{align}\label{Design:NCS}
\dot{z}(t)=G^{-1}AG z(t)+G^{-1}BK_c z(t_k-d),~~~~t\in[t_k,t_{k+1})
\end{align}
where $K_c:=KG$.

System \eqref{Design:NCS} exhibits the same stability behavior as system \eqref{sys:sampling}, and the triggering condition \eqref{sys:trigger}
is still valid.
Based on Theorem \ref{Th1} and the equivalent system expression in \eqref{Design:NCS}, we have the following stability result under Assumption \ref{Ass:matrix:AB}.
\begin{Theorem}[\emph {Co-designing under unknown $A$ and $B$}]\label{Th:Co-designing:AB}
For given scalars $\bar{h}>\underline{h}>0$, $\bar{d}>\underline{d}>0$, $\sigma_{1}>0$, $\sigma_{2}>0$, {\color{blue}$\lambda>0$, and $\theta>0$ satisfying $1-\lambda-\frac{1}{\theta}\geq0$, or $\theta=0$, there exists a controller gain $K$ such that system \eqref{sys:sampling} is asymptotically stable under the triggering condition in \eqref{sys:trigger} for any $[A ~B]\in \Sigma_s$,
and $\eta(\tau_j-d)$ converges to the origin for any $\eta(0)\geq 0$}, if there exist scalars $\varepsilon>0$, $\epsilon>0$, and matrices $P\succ0$, $Z\succ0$, $T\succ0$, {\color{blue}$R_1\succ0$, $R_2\succ0$, $\Omega\succ0$,
$S=S^T$, $N$, $M_1$, $M_2$,} $G$, $K_c$ such that the following LMIs hold for all $h\in\{\underline{h}, \bar{h}\}$ and $d\in\{\underline{d}, \bar{d}\}$
{\color{blue}
\begin{align}{\label {Th:controller:AB:LMI1}}
&\left[\!\!
  \begin{array}{cccc}
    \mathcal{B}_1& \mathcal{B}_2+\big[GL_1^T~K_cL_{10}^T\big]^T& 0\\
    \ast & \mathcal{B}_3+\Xi_0+h\Xi_a+\check{\Psi}+\mathcal{O}  & dN \\
    \ast & \ast & d\mathcal{T} \\
    \ast & \ast & \ast
  \end{array}
\!\!\right]\!\!\prec0
\end{align}
\begin{align}{\label {Th:controller:AB:LMI2}}
&\left[\!\!
 \begin{array}{cccc}
     \mathcal{B}_1& \mathcal{B}_2+\big[GL_1^T~K_cL_{10}^T\big]^T& 0& 0\\
    \ast & \mathcal{B}_3+\Xi_{0}+h\Xi_b+\check{\Psi}+\mathcal{O}  & dN & h\mathcal{M}\\
    \ast & \ast & d\mathcal{T} &0\\
    \ast & \ast & \ast & h\mathcal{R}
  \end{array}
  \!\!\right]\prec0\end{align}}
where
\begin{align*}
\check{\Psi}&:={\rm Sym}\big\{-(L_1^T+\epsilon L_3^T)GL_{3}\big\}\\
\mathcal{B}_1&:=\varepsilon\mathcal{V}_1\Theta_s\mathcal{V}_1^T,~
\mathcal{B}_2:=\varepsilon\mathcal{V}_1\Theta_s\mathcal{V}_2^T,~
\mathcal{B}_3:=\varepsilon\mathcal{V}_2\Theta_s\mathcal{V}_2^T\\
\mathcal{V}_1&:=
\left[\begin{array}{ccc}I & 0& 0\\
0 & 1 & 0\\\end{array}\right],~
\mathcal{V}_2:=
\left[\begin{array}{ccc}0 & 0& (L_1^T+\epsilon L_3^T)\\\end{array}\right].
\end{align*}
Moreover, the controller gain $K$ is given by $K=K_cG^{-1}$.
\end{Theorem}
\begin{proof}
We construct a functional $V(z,t)$ for system \eqref{Design:NCS} by substituting $x$ of the functional $V(x,t)$ in \eqref{Th1:Vt} with $z$.
Similar to the proof of Theorem \ref{Th1},
it can be obtained that
 \begin{align}\label{Th4:Vd}
\dot V(z,t)\leq&~ \xi^{T}(z,t)\left[\frac{\tau_{j+1}-t}{h}\Upsilon_1(h)
+\frac{t-\tau_j}{h}\Upsilon_2(h) \right]\xi(z,t).
\end{align}

According to system \eqref{Design:NCS}, we have that
\begin{align}{\label {Th:controller:AB:zero}}
0=&2\big[z(t)+\epsilon \dot z(t)\big]\times \notag\\
&\times G\big[G^{-1}AG z(t)+G^{-1}BK_cz(t_k-d)-\dot z(t) \big].
\end{align}
Replacing the equality \eqref{Th1:zero} by \eqref{Th:controller:AB:zero} in the proof of Theorem \ref{Th1}, it can be obtained that
 \begin{align}\label{Th4:Vd}
\dot V(z,t)\leq&~ \xi^{T}(z,t)\left[\frac{\tau_{j+1}-t}{h}\check{\Upsilon}_1(h)
+\frac{t-\tau_j}{h}\check{\Upsilon}_2(h) \right]\xi(z,t).
\end{align}
where $\check{\Upsilon}_i(h):={\rm Sym}\big\{(L_1^T+\epsilon L_3^T)(AGL_1+BK_cL_{10})\big\}+{\Upsilon}_i(h)-{\Psi}+\check{\Psi}$ for $i=1,2$.
Then, terms $\check{\Upsilon}_1(h)$ and $\check{\Upsilon}_2(h)$ are rewritten as
\begin{align}\label{Th4:decomposition1}
&\check{\Upsilon}_i(h)=
\left[\!\begin{array}{cc}(L_1^T+\epsilon L_3^T)[A~B]^T\\I \\\end{array}\!\right]^T
\left[\!\begin{array}{cc}0 \!&\! \big[GL_1^T~K_cL_{10}^T\big]^T\\\ast \!&\! \Upsilon_i(h)-\Psi +\check{\Psi} \\\end{array}\!\right]
[\cdot]
\end{align}

According to the data-based representation in \eqref{data:represent}, it holds that
\begin{align}
\left[\!\begin{array}{cc}[A~B]^T\\I \\\end{array}\!\right]^T
  \!\Theta_s\!
  \left[\!\begin{array}{cc}[A~B]^T \\
   I\\\end{array}\!\right]\succeq0.
\end{align}

By the full-block S-procedure, we have $\check{\Upsilon}_1(h)\prec0$ and $\check{\Upsilon}_2(h)\prec0$ for any $[A ~B]\in\Sigma_s$ if there exists a scalar $\varepsilon>0$ such that
 \begin{align}\label{Th4:fullblock1}
&\left[\begin{array}{cc}0 & \big[GL_1^T~K_cL_{10}^T\big]^T\\\ast & \Upsilon_1(h)-\Psi +\check{\Psi} \\\end{array}\right]+
\varepsilon \left[\begin{array}{cc}\mathcal{V}_1\Theta_s\mathcal{V}_1^T & \mathcal{V}_1\Theta_s\mathcal{V}_2^T\\\ast &  \mathcal{V}_2\Theta_s\mathcal{V}_2^T \\\end{array}\right]\prec0 \\ \label{Th4:fullblock2}
&\left[\begin{array}{cc}0 & \big[GL_1^T~K_cL_{10}^T\big]^T\\\ast & \Upsilon_2(h)-\Psi +\check{\Psi} \\\end{array}\right]+
\varepsilon \left[\begin{array}{cc}\mathcal{V}_1\Theta_s\mathcal{V}_1^T & \mathcal{V}_1\Theta_s\mathcal{V}_2^T\\\ast &  \mathcal{V}_2\Theta_s\mathcal{V}_2^T \\\end{array}\right]\prec0.
\end{align}

Finally, similar to the proof of Theorem \ref{Th2}, we have a conclusion that \eqref{Th:controller:AB:LMI1} and \eqref{Th:controller:AB:LMI2} are sufficient stability conditions for System \eqref{Design:NCS}
under the triggering condition \eqref{sys:trigger} for any $[A ~B]\in \Sigma_s$, and $\eta(\tau_j-d)$ converges to the origin.
Since $G$ is a nonsingular matrix, system \eqref{Design:NCS} exhibits the same stability behavior as system \eqref{sys:sampling}, which completes the proof.
\end{proof}

\begin{Remark}[\emph {Novelty}] \label{novelty}
{\color{blue}
The design procedure of Theorem \ref{Th:Co-designing:AB} requires no model knowledge, similar to \cite{Berberich2020}, but only a single open-loop data trajectory affected by noise.}
Note that LMIs \eqref{Th:controller:AB:LMI1} and \eqref{Th:controller:AB:LMI2} are convex in all decision variables, in contrast to Theorems \ref{Th2}-\ref{Th:Co-designing:AB:F} and \cite[Theorems 4-5]{Berberich2020}.
We can obtain a controller gain $K$ with an allowable MSI directly by solving the LMIs in Theorem \ref{Th:Co-designing:AB} without fixing any variables. In essence, Theorem \ref{Th:Co-designing:AB} is established at the price of increasing the conservatism of free matrix $F$; that is, replacing the equality \eqref{Th1:zero} with \eqref{Th:controller:AB:zero} in the proof of Theorem \ref{Th1}.
As a result, the free matrix $F$ of dimensions $10n\times n$ is limited to the matrix $G$ of $n\times n$. That is to say, Theorem \ref{Th:Co-designing:AB} is more conservative than Theorems \ref{Th2}-\ref{Th:Co-designing:AB:F}. 
This conclusion is further corroborated through a numerical comparison in Table \ref{Tab:noise:k}.
Besides,
in contrast to Theorems \ref{Th2}--\ref{Th3} and \cite[Theorem 4]{Berberich2020}, the restriction of requiring an invertible matrix $\Theta_s$ (cf. Assumption \ref{Ass:matrix}) is removed in Theorem \ref{Th:Co-designing:AB}.
 The main reason is that the terms in \eqref{Th4:decomposition1}
 match the inequality \eqref{data:represent} of Lemma \ref{Lemma:system:data} exactly.
 We then do not need the intermediate step of computing $\Theta_s^{-1}$ (namely, \eqref{Th2:dualization} in Theorem \ref{Th2}) when designing the controller gain $K$ by means of Theorem \ref{Th:Co-designing:AB}.
 Thus, this incurs lower computational complexity and no restriction (Assumption \ref{Ass:matrix}) in co-designing the controller gain and the triggering matrix $\Omega$.
\end{Remark}

Next, we consider using the prior knowledge of input matrix $B$ for obtaining a larger MSI of system \eqref{sys:sampling} under the triggering condition \eqref{sys:trigger} in the co-design method of Theorem \ref{Th:Co-designing:AB}. For this purpose,
a co-design method of the controller gain $K$ and the triggering matrix $\Omega$ for system \eqref{sys:sampling} with an \emph {unknown} state matrix $A$ and any \emph {known} input matrix $B$ (cf. Assumption \ref{Ass:matrix:A}) is given based on Theorem \ref{Th:Co-designing:AB}.

\begin{Theorem}[\emph {Co-designing under unknown $A$ and known $B$}]\label{Th:controller:A}
For given scalars $\bar{h}>\underline{h}>0$, $\bar{d}>\underline{d}>0$, $\sigma_{1}>0$, $\sigma_{2}>0$, {\color{blue}$\lambda>0$, and $\theta>0$ satisfying $1-\lambda-\frac{1}{\theta}\geq0$, or $\theta=0$, there exists a controller gain $K$ such that system \eqref{sys:sampling} is asymptotically stable under the triggering condition in \eqref{sys:trigger} for any $[A]\in \bar{\Sigma}_s$,
and $\eta(\tau_j-d)$ converges to the origin for any $\eta(0)\geq 0$}, if there exist scalars $\varepsilon>0$, $\epsilon>0$, and matrices $P\succ0$, $Z\succ0$, $T\succ0$, {\color{blue}$R_1\succ0$, $R_2\succ0$, $\Omega\succ0$,
$S=S^T$, $N$, $M_1$, $M_2$,} $G$, $K_c$ such that the following LMIs hold for $h\in\{\underline{h}, \bar{h}\}$ and $d\in\{\underline{d}, \bar{d}\}$
{\color{blue}\begin{align}{\label {Th:controller:A:LMI1}}
&\left[\!\!
  \begin{array}{cccc}
  \grave{ \mathcal{B}}_1& \grave{\mathcal{B}}_2+Gl_1& 0\\
    \ast & \grave{\mathcal{B}}_3+\Xi_{0}+h\Xi_a+\grave{\Psi}+\mathcal{O}  & dN \\
    \ast & \ast & d\mathcal{T} \\
    \ast & \ast & \ast
  \end{array}
\!\!\right]\!\!\prec0
\end{align}
\begin{align}{\label {Th:controller:A:LMI2}}
&\left[\!\!
 \begin{array}{cccc}
     \grave{\mathcal{B}}_1& \grave{\mathcal{B}}_2+Gl_1& 0& 0\\
    \ast & \grave{\mathcal{B}}_3+\Xi_{0}+h\Xi_b+\grave{\Psi}+\mathcal{O}  & dN & h\mathcal{M}\\
    \ast & \ast & d\mathcal{T} &0\\
    \ast & \ast & \ast & h\mathcal{R}
  \end{array}
  \!\!\right]\prec0\end{align}}
where
\begin{align*}
\grave{\Psi}&:={\rm Sym}\big\{(L_1^T+\epsilon L_3^T)(BK_cl_{19}-Gl_{18})\big\}\\
\grave{\mathcal{B}}_1&:=\varepsilon\grave{\mathcal{V}}_1\Theta_s\grave{\mathcal{V}}_1^T,~
\grave{\mathcal{B}}_2:=\varepsilon\grave{\mathcal{V}}_1\Theta_s\grave{\mathcal{V}}_2^T,~
\grave{\mathcal{B}}_3:=\varepsilon\grave{\mathcal{V}}_2\Theta_s\grave{\mathcal{V}}_2^T\\
\grave{\mathcal{V}}_1&:=
\left[\begin{array}{ccc}I & 0\\\end{array}\right],~
\grave{\mathcal{V}}_2:=
\left[\begin{array}{ccc} 0& (L_1^T+\epsilon L_3^T)\\\end{array}\right].
\end{align*}
Moreover, the controller gain $K$ is given by $K=K_cG^{-1}$.
\end{Theorem}

\begin{proof} The whole proof is similar to Theorem \ref{Th:Co-designing:AB}, except that
we re-express terms $\check{\Upsilon}_1(h)$ and $\check{\Upsilon}_2(h)$ as, for $i=1,2$
\begin{align}\label{Th5:decomposition1}
&\check{\Upsilon}_i(h)=
\left[\begin{array}{cc}[(L_1^T+\epsilon L_3^T)A]^T\\I \\\end{array}\right]^T
\left[\begin{array}{cc}0 & Gl_1\\\ast & \Upsilon_i(h)-\Psi +\grave{\Psi} \\\end{array}\right][\cdot]
\end{align}

Then, from the data-based representation \eqref{data:represent:K}, it holds that
\begin{align}
\left[\!\begin{array}{cc}A^T\\I \\\end{array}\!\right]^T
  \!\Theta_s\!
  \left[\!\begin{array}{cc}A^T \\
   I\\\end{array}\!\right]\succeq0.
\end{align}

Through the full-block S-procedure, it holds that $\check{\Upsilon}_1(h)\prec0$ and $\check{\Upsilon}_2(h)\prec0$  for
any $[A]\in \bar{\Sigma}_s$ if there exists a scalar $\varepsilon>0$ such that
 \begin{align}\label{Th5:fullblock1}
&\left[\begin{array}{cc}0 & Gl_1\\\ast & \Upsilon_1(h)-\Psi +\grave{\Psi} \\\end{array}\right]+
\varepsilon \left[\begin{array}{cc}\grave{\mathcal{V}}_1\Theta_s\grave{\mathcal{V}}_1^T & \grave{\mathcal{V}}_1\Theta_s\grave{\mathcal{V}}_2^T\\\ast &  \grave{\mathcal{V}}_2\Theta_s\grave{\mathcal{V}}_2^T \\\end{array}\right]\prec0 \\ \label{Th5:fullblock2}
&\left[\begin{array}{cc}0 & Gl_1\\\ast & \Upsilon_2(h)-\Psi +\grave{\Psi} \\\end{array}\right]+
\varepsilon \left[\begin{array}{cc}\grave{\mathcal{V}}_1\Theta_s\grave{\mathcal{V}}_1^T & \grave{\mathcal{V}}_1\Theta_s\grave{\mathcal{V}}_2^T\\\ast &  \grave{\mathcal{V}}_2\Theta_s\grave{\mathcal{V}}_2^T \\\end{array}\right]\prec0.
\end{align}

Finally, similar to Theorem \ref{Th:Co-designing:AB}, we conclude that \eqref{Th:controller:A:LMI1} and \eqref{Th:controller:A:LMI2} are sufficient stability conditions for system \eqref{sys:sampling} under the triggering condition \eqref{sys:trigger} for any $[A]\in \bar{\Sigma}_s$, and $\eta(\tau_j-d)$ converges to the origin.
\end{proof}

Theorem \ref{Th:controller:A} provides a co-design method for obtaining matrices $K$ and $\Omega$. Compared
to Theorem \ref{Th:Co-designing:AB}, which is based only on the data, a larger MSI bounds may be computed by
 Theorem \ref{Th:controller:A}; see the numerical comparison between Theorems \ref{Th:Co-designing:AB} and \ref{Th:controller:A}
for different noise bounds given in Table \ref{Tab:noise:k}.

{\color{blue}
\begin{Remark}[\emph {Alternative approach}]
Sampled-data control systems, in which information is transmitted over wired or wireless networks having limited bandwidth \cite{Zhang2001}, have been widely studied over the last two decades. Event-triggered control has emerged as a paradigm for sampled-data systems \cite{Heemels2008}. In existing works on event-triggered control, exact knowledge of a plant is assumed. However,  in practical situations, it is often difficult if not impossible to acquire an accurate system model.
This motivates well the combination of data-driven control and event-triggering control for  sampled-data systems. In our proposed approach, the controller and the triggering protocol are simultaneously designed without explicit knowledge of the system model.
To the authors' knowledge, the only alternative to our approach in the current literature would be first performing system identification, using e.g., least-squares estimation of the system matrices, followed by event-triggered control that has been investigated in \cite{Yue2013,Liu2018}.
However, such an identification-based event-triggered control does in general not provide stability guarantees, in particular when the data are affected by noise.
This is due to the fact that i) providing tight estimation bounds for system identification is  challenging, compare  e.g., \cite{Matni2019}; and ii) \cite{Yue2013,Liu2018} only provide nominal results, i.e., potential error bounds, e.g., arising from system identification are not handled systematically.
In contrast, our event-triggered control approach admits rigorous stability guarantees based directly on noisy data. 
\end{Remark}}

\section{Example and Simulation}\label{sec:example}
\textcolor{blue}{Two numerical examples are presented in this section to demonstrate the effectiveness of our proposed methods.}

\begin{Example}\label{ex:2class}
\emph{Consider system \eqref{sys:sampling} under the transmission scheme \eqref{sys:trigger} with parameters in \cite{Zhang2001}
\begin{equation*}
A=\left[
  \begin{array}{cc}
    0& 1 \\
    0 & -0.1 \\
  \end{array}
\right],~~
B=\left[
  \begin{array}{cc}
    0 \\
    0.1\\
  \end{array}
\right].
\end{equation*}
\textcolor{blue}{Assume that the matrices $A$ and $B$ are \emph{unknown}. We consider data-driven event-triggered control of system \eqref{sys:sampling}. To this end, we generate measurements $\{\dot x(T_i),\,x(T_i),\,u(T_i)\}^{100}_{i=1}$ with the interval $T_{i+1}-T_i$ satisfying
\begin{equation}\label{example:sampling}
T_{i+1}-T_i=\left\{\begin{array}{lll}
1&,&t\in [1,\, 49]\\
2&, &t\in [50,\, 99]
\end{array}
\right.
\end{equation}
where the data-generating input was sampled uniformly from $u(t)\in [-1,~1]$.
The measured data were perturbed by a disturbance distributed uniformly over $w(t)\in [-\bar{w},\,\bar{w}]$ for $\bar{w}>0$. According to Remark \ref{remark:noise}, $w(t)$ satisfies Assumption \ref{Ass:disturbance} with $Q_d=-I$, $S_d=0$, and $R_d=\bar{w}^2\rho I$ $(\rho=100)$.}}
\end{Example}

%
\begin{table}[t]
{\color{blue}
\caption{Maximum allowable $\bar{h}$ with $\underline{h}=10^{-5}$ and delay $\underline{d}=0$ for different bounds $\bar{w}$, $\bar{d}$, and given $K=-[3.75~ 11.5]$.}
\begin{center}      
\setlength{\tabcolsep}{6pt}
\begin{tabular}{lccccccccc}
\hline\noalign{\smallskip}
$\bar{w}$ & 0.005 & 0.01 & 0.02 & 0.03&0.04& 0.05\\
\noalign{\smallskip}\hline\noalign{\smallskip}
Theorem \ref{Th2} ($\bar{d}=0$)&1.04& 0.96&  0.81&  0.72&   0.64&  0.57\\
Theorem \ref{Th3} ($\bar{d}=0$)&1.12&  1.09& 1.04& 1.00& 0.99&0.91\\
\noalign{\smallskip}\hline\noalign{\smallskip}
Theorem \ref{Th2} ($\bar{d}=0.1$)&0.78&  0.70&  0.58&  0.47&   0.38&  0.31\\
Theorem \ref{Th3} ($\bar{d}=0.1$)&0.86&  0.84& 0.82& 0.76& 0.73& 0.71\\
\noalign{\smallskip}\hline\noalign{\smallskip}
Theorem \ref{Th2} ($\bar{d}=0.2$)&0.65& 0.57& 0.44& 0.35&   0.27&  0.21\\
Theorem \ref{Th3} ($\bar{d}=0.2$)&0.73&  0.71& 0.68& 0.64& 0.61& 0.59\\
\noalign{\smallskip}\hline
\end{tabular}
\end{center}
\label{Tab:noise}}
\end{table}
\textcolor{blue}{(\emph{MSI bounds})}
Using Theorem \ref{Th2} and controller gain $K=-[3.75~ 11.5]$, possibly large values of $\bar{h}$ with $\underline{h}=10^{-5}$ and $\underline{d}=0$ for different realizations of $\bar{w}$ and $\bar{d}$ were computed and are presented in Table \ref{Tab:noise}.
It can be observed that larger values of $\bar{w}$ lead to a higher sampling frequency to guarantee stability, which is consistent with \cite{Berberich2020}. {\color{blue}Besides, from Table \ref{Tab:noise}, the increase of transmission delay reduces the MSI bound, i.e., more transmissions are required for stability.
}
Furthermore, we computed possibly large bounds $\bar{h}$ for different $\bar{w}$ by using Theorem \ref{Th3}, i.e., assuming that the input matrix B is known.
As shown in Table \ref{Tab:noise}, Theorem \ref{Th3} provides larger values of $\bar{h}$ than Theorem \ref{Th2}, which confirms that additional prior knowledge may lead to less
conservative data-based stability conditions compared with those established solely on the available data.

In Table \ref{Tab:iteration}, values of $\bar{h}$ associated with the  controller gain $K$ iteratively found by \cite[Theorems 4-5]{Berberich2020}, Theorems \ref{Th2} and \ref{Th:Co-designing:AB:F}, as well as Theorem \ref{Th3} are displayed. There is a similar conclusion as in Table \ref{Tab:noise} that
the iterative approach using Theorem \ref{Th3}
enhances robustness to noise relative to \cite[Theorems 4-5]{Berberich2020} and Theorem \ref{Th2} with Theorem \ref{Th:Co-designing:AB:F}.
Except for this, by choosing a different controller gain, improved MSI bounds can be obtained by
the iterative approaches in Table \ref{Tab:iteration} when compared to the ones in Table \ref{Tab:noise} obtained by the fixed controller gain.
Table \ref{Tab:noise:k} depicts
values of $\bar{h}$ found by Theorems \ref{Th:Co-designing:AB}-\ref{Th:controller:A} with $\epsilon=2$. 
 Theorems \ref{Th:Co-designing:AB} and \ref{Th:controller:A} provide smaller $h$ values than Theorems \ref{Th2}-\ref{Th3}, which confirms  usefulness of the free matrix $F$ in reducing stability conservatism.
It is worth noting that,
at the price of additional conservatism with respect to matrix $F$, Theorems \ref{Th:Co-designing:AB} and \ref{Th:controller:A} offer convex LMI conditions. Different from Theorems \ref{Th2}-\ref{Th:Co-designing:AB:F} and \cite[Theorems 4-5]{Berberich2020}, matrices $K$ and $\Omega$ can be co-designed directly based on Theorems \ref{Th:Co-designing:AB} and \ref{Th:controller:A}.

\begin{table}[t]
\color{blue}
\caption{Maximum allowable $\bar{h}$ with $\underline{h}=10^{-5}$ and delay $d=0$ for different noise bounds $\bar{w}$ by iterative approaches.}
\begin{center}      
\begin{tabular}{lccccccccc}
\hline\noalign{\smallskip}
$\bar{w}$ & 0.005 & 0.01 & 0.02 & 0.03&0.04& 0.05\\
\noalign{\smallskip}\hline\noalign{\smallskip}
Theorems \ref{Th2} and \ref{Th:Co-designing:AB:F}& 23.4  & 14.7 &5.2 & 3.4 &2.1 & 1.9\\
\cite[Theorems 4-5]{Berberich2020}&28.5&13.8&6.3&4.0&2.9&2.2\\
Theorem \ref{Th3} &32.2 & 18.8 & 8.1 &   4.8  &  3.5  &  2.6\\
\noalign{\smallskip}\hline
\end{tabular}
\end{center}
\label{Tab:iteration}
\end{table}

\begin{table}[t]
\color{blue}
\caption{Maximum allowable $\bar{h}$ with $\underline{h}=10^{-5}$ and delay $d=0$ for different noise bounds $\bar{w}$ and the optimized controller gain $K$.}
\begin{center}      
\begin{tabular}{lccccccccc}
\hline\noalign{\smallskip}
$\bar{w}$ & 0.005 & 0.01 & 0.02 & 0.03&0.04& 0.05\\
\noalign{\smallskip}\hline\noalign{\smallskip}
Theorem \ref{Th:Co-designing:AB}&11.6&  8.4&  4.2&   2.7&  1.9 & 1.6\\
Theorem \ref{Th:controller:A}&14.0&  10.9& 5.1&  3.2&  2.0&  1.8\\
\noalign{\smallskip}\hline
\end{tabular}
\end{center}
\label{Tab:noise:k}
\end{table}

\textcolor{blue}{(\emph{Co-designing under unknown $A$ and $B$})} For data-driven control of system \eqref{sys:sampling} with unknown matrices $A$ and $B$ under our dynamic transmission scheme \eqref{sys:trigger}, we employ Theorem \ref{Th:Co-designing:AB} to co-design the controller gain $K$ and
triggering matrix $\Omega$. We consider constant transmission {\color{blue}delay bounds $\underline{d}=0$, $\bar{d}=0.1$}, sampling interval $h=0.2$, triggering parameters $\sigma_1=0.4$, $\sigma_2=0.1$, $\theta=2$, $\lambda=0.5$, $\epsilon=2$, noise $w(t)\in [-0.01,~0.01]$, and the available measurements $\{\dot x(T_i),~x(T_i),~u(T_i)\}^{100}_{i=1}$ with the interval $T_{i+1}-T_i$ satisfying \eqref{example:sampling}.
Solving the data-based LMIs \eqref{Th:controller:AB:LMI1} and \eqref{Th:controller:AB:LMI2},
the controller gain and the triggering matrix were co-designed as follows
{\color{blue}\begin{align*}
\Omega&=\left[
  \begin{array}{cccc}
   218.2530 &  -59.7385\\
   -59.7385  & 95.9565\\
  \end{array}
\right]\\
K&=\left[
  \begin{array}{cccc}
    -0.1798 &  -3.7008
  \end{array}
\right].
\end{align*}}
We then simulated system \eqref{sys:sampling} under the triggering scheme in \eqref{sys:trigger} with initial condition $x(0) = [3,\, -2]^T$, as well as the dynamic variable $\eta(\tau_j-d)$, over the time interval $[0,30]$. Their trajectories are depicted in Fig. \ref{FIG:xt:AB}.
Evidently, both the system states and the dynamic variable converge to the origin,
which demonstrates the feasibility of our designed controller gain $K$ and the triggering matrix $\Omega$. Additionally, Fig. \ref{FIG:xt:AB} plots the evolutions of error function $\iota f_{c}$ and threshold $\iota f_{t}$, where $\iota:=\tau_j^3$ is a weighting function for visualizing triggers, and
\textcolor{blue}{\begin{align*}\label{threshold}
f_{c}&:=e^T(\tau_j-d)\Omega e(\tau_j-d),\\
f_{t}&:=\frac{1}{\theta}\eta(\tau_j-d)+\sigma_1 x^T(\tau_j-d)\Omega x(\tau_j-d)\\
&~~~~+\sigma_2 x^T(t_k-d)\Omega x(t_k-d).
\end{align*}}
It is evident from the plots that an event is triggered as soon as the triggering condition $\iota f_{c}$ exceeds
the threshold $\iota f_{t}$. \textcolor{blue}{At that triggering instant, the state $x(\tau_j-d)$ is transmitted to the controller, and the measurement error $e(\tau_j-d)$ is set to be zero in $\eta(\tau_j-d)$. Thus, consistent with the theoretical result in \eqref{Lemma1:1}, the error $\iota f_{c}$ is kept smaller than the threshold $\iota f_{t}$. }
It is also worth pointing out that only $20$ measurements were transmitted to the controller under our proposed triggering scheme \eqref{sys:trigger}, while $150$ measurements were sampled.
This validates the effectiveness of the proposed scheme in saving communication resources, while maintaining system stability.
\begin{figure}[t]\color{blue}
\centering
\subfigure{
\includegraphics[scale=0.6]{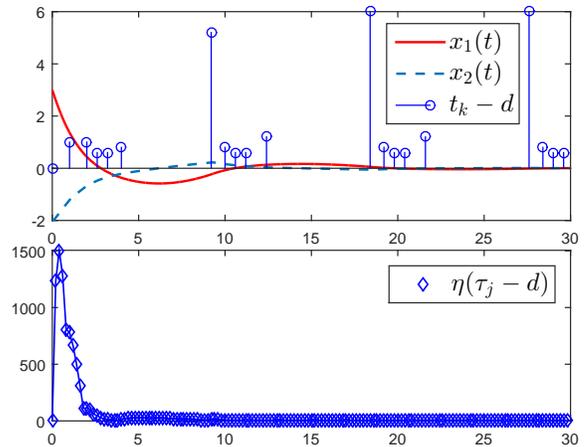}
}
\quad
\subfigure{
\includegraphics[scale=0.6]{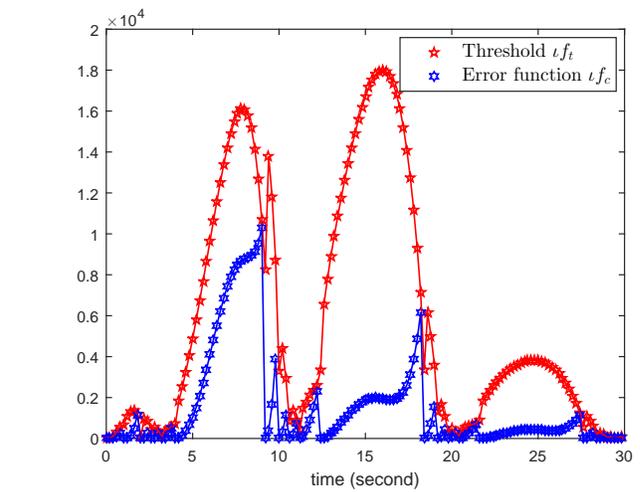}
}
\caption{Trajectories of system \eqref{sys:sampling}, dynamic variable $\eta(\tau_j-d)$, error function $f_{c}$, and threshold $f_{t}$ under triggering condition \eqref{sys:trigger} with parameters
$h=0.2$, $\sigma_1=0.4$, $\sigma_2=0.1$, $\theta=1$, $\lambda=0.5$, unknown matrices $A$ and $B$ with initial condition $x(0) = [3,\, -2]^T$, and delay $d=0.1$.}
\label{FIG:xt:AB}
\end{figure}

\begin{figure}[t]\color{blue}
\centering
\subfigure{
\includegraphics[scale=0.6]{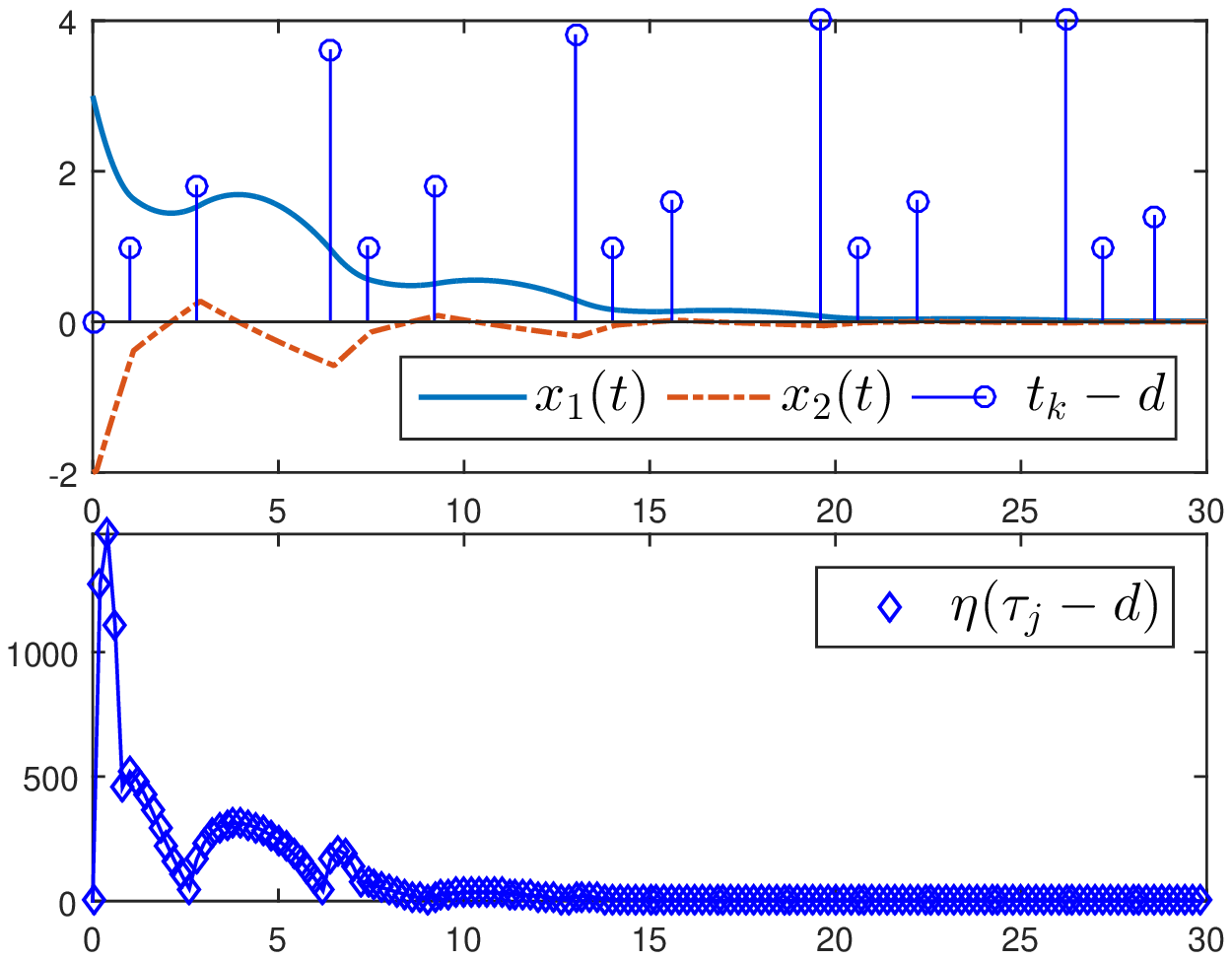}
}
\subfigure{
\includegraphics[scale=0.6]{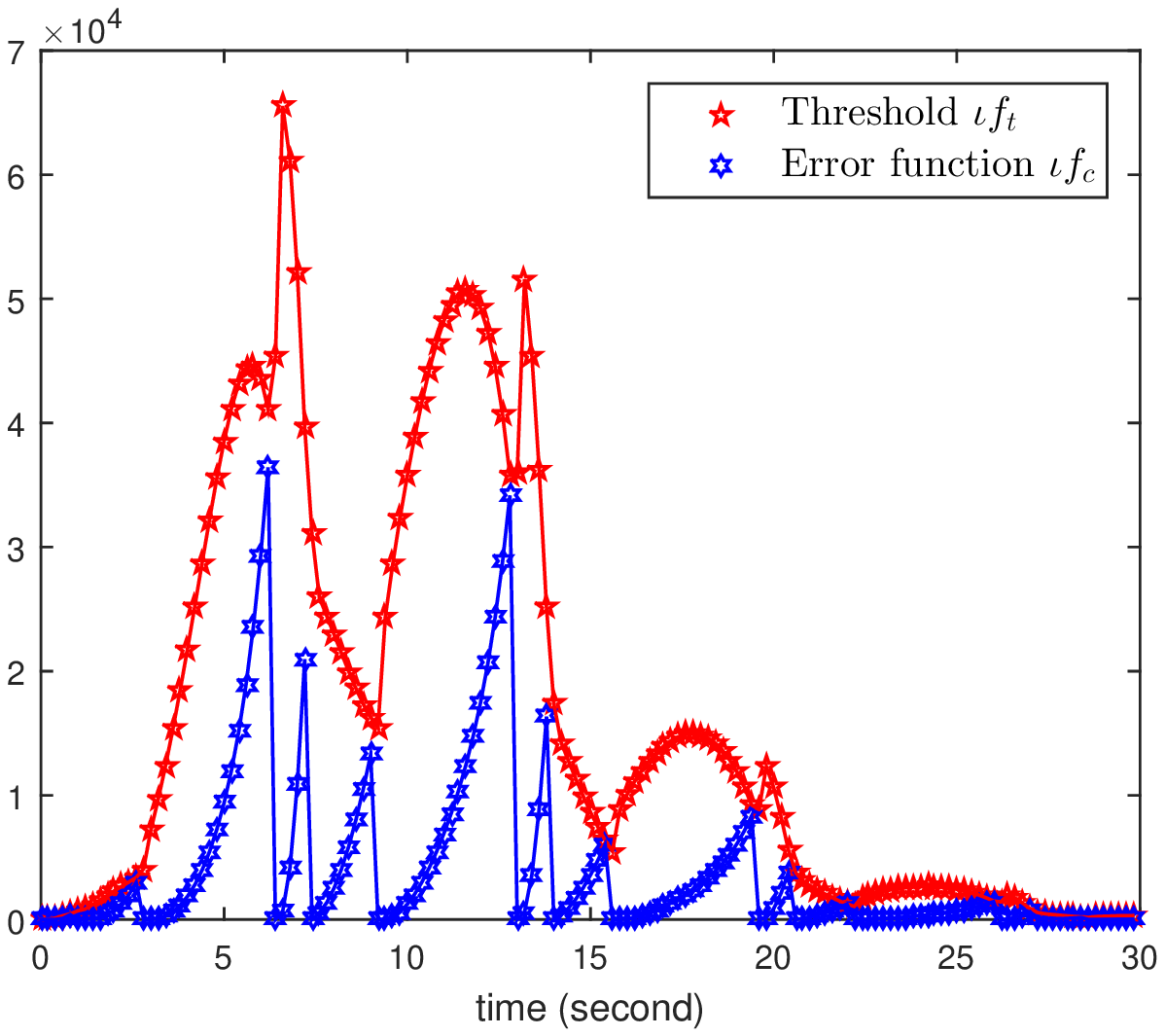}
}
\caption{Trajectories of system \eqref{sys:sampling}, dynamic variable $\eta(\tau_j-d)$, error function $f_{c}$, and threshold $f_{t}$ under triggering condition \eqref{sys:trigger} with parameters
$h=0.2$, $\sigma_1=0.4$, $\sigma_2=0.1$, $\theta=2$, $\lambda=0.5$, unknown matrix $A$ and known $B$ with initial condition $x(0) = [3,\, -2]^T$, and delay $d=0.1$.}
\label{FIG:xt:dy:A}
\end{figure}

\textcolor{blue}{(\emph{Co-designing under unknown $A$ and known $B$})} For the event-triggered system \eqref{sys:sampling} with unknown $A$ and known $B$ under the dynamic transmission scheme \eqref{sys:trigger},
we considered the same parameters as used in Fig. \ref{FIG:xt:AB}
and the available measurements $\{\dot x(T_i),~x(T_i),~u(T_i)\}^{100}_{i=1}$ with the interval $T_{i+1}-T_i$ satisfying \eqref{example:sampling}.
Solving the LMIs
\eqref{Th:controller:A:LMI1} and \eqref{Th:controller:A:LMI2},
the co-designed controller gain and triggering matrix leveraging Theorem \ref{Th:controller:A}
were computed as follows
{\color{blue}
\begin{align*}
\Omega&=\left[
  \begin{array}{cccc}
   135.0611 & -65.0218\\
   -65.0218 & 292.8338\\
  \end{array}
\right]\\
K&=\left[
  \begin{array}{cccc}
  -0.4101 &  -8.0184
  \end{array}
\right].
\end{align*}}
The simulation results are shown in Fig. \ref{FIG:xt:dy:A}.
Both the system states and the dynamic variable converged to the origin, while 15 measurements out of 150 sampled-data  were transmitted to the controller. Clearly, the co-design method in Theorem \ref{Th:controller:A} is effective for System \eqref{sys:sampling} with unknown system matrix $A$ and known matrix $B$.
Last but not least, by comparing Figs. \ref{FIG:xt:AB} and \ref{FIG:xt:dy:A}, the knowledge of $B$ allows to stabilize the system using fewer transmissions.

\begin{figure}[t]\color{blue}
\centering
\subfigure{
\includegraphics[scale=0.6]{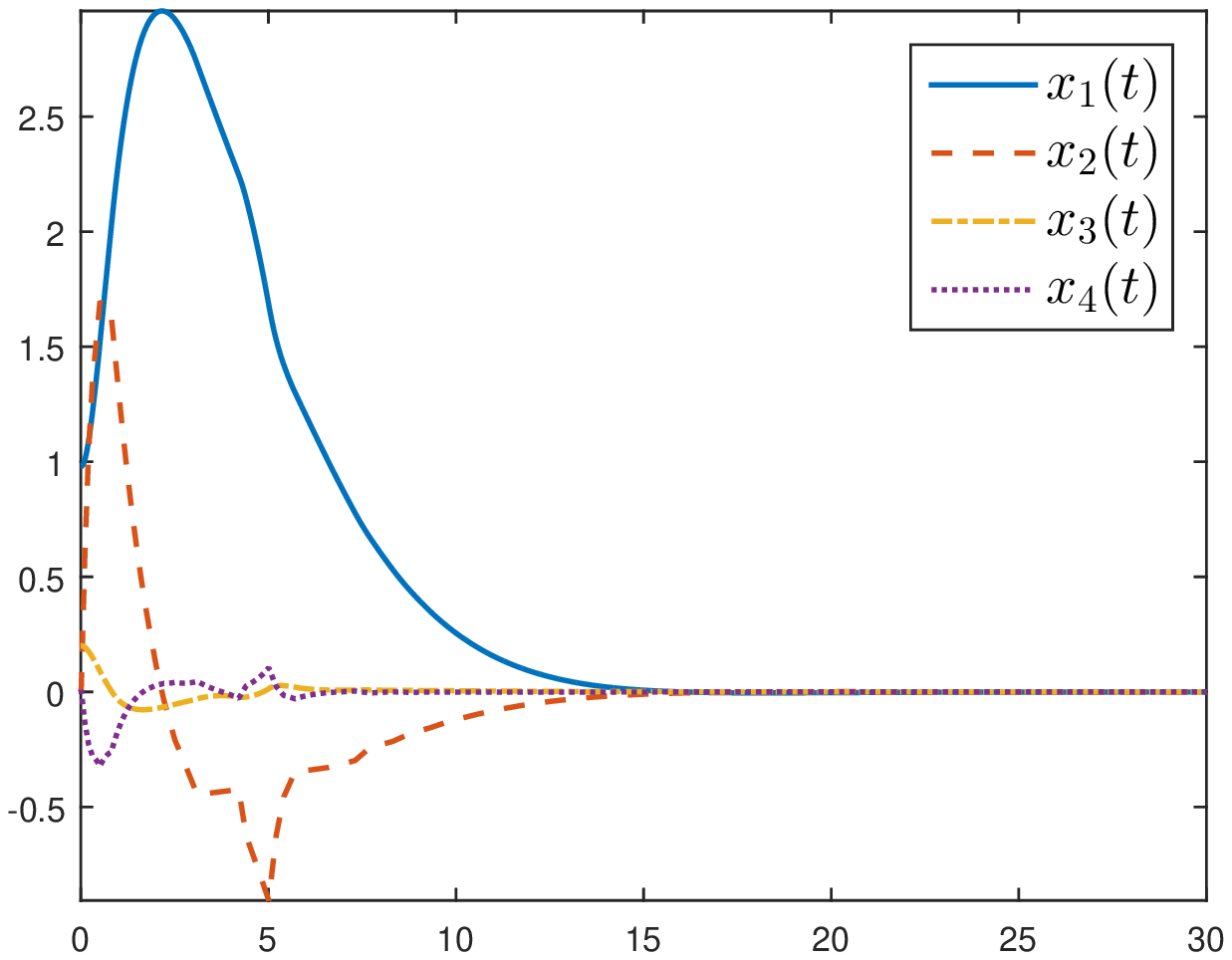}
}
\subfigure{
\includegraphics[scale=0.6]{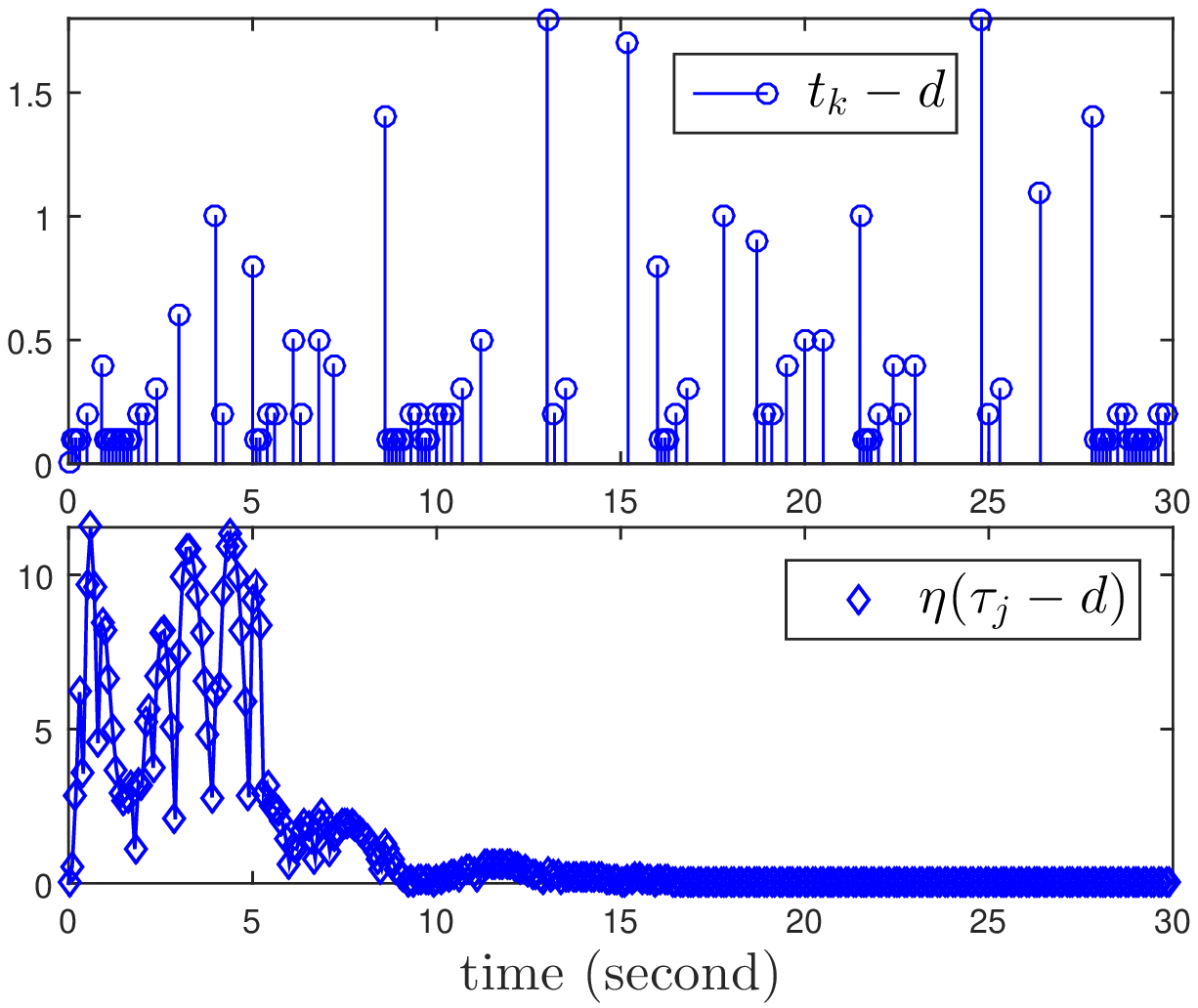}
}
\caption{Trajectories of system \eqref{sys:sampling}, dynamic variable $\eta(\tau_j)$ under triggering condition \eqref{sys:trigger} with parameters
$h=0.1$, $\sigma_1=0.01$, $\sigma_2=0.01$, $\theta=2$, $\lambda=0.5$, unknown matrices $A$ and $B$ with initial condition $x(0) = [0.98,~0,~ 0.2,~0]^T$, and delay $d=0.01$.}
\label{FIG:xt3}
\end{figure}
{\color{blue}
\begin{Example}
\emph{The second example involves a fourth-order system to further demonstrate the effectiveness of the proposed
 data-driven event-triggering scheme.
Consider a linearized inverted pendulum system \cite{Yue2013}, whose state-space representation is given by}
	\begin{equation*}
		\dot x(t)\!=\!\left[
		\begin{array}{cccc}
			\!0 & 1 & 0 & 0\! \\
			\!0 & 0 & {-m_1 g}/{m_2}& 0 \!\\
			\!0 & 0 & 0 & 1\!\\
		    \!0 & 0 & {g}/{l} & 0\!\\
				\end{array}
		\right]\!x(t)\!+\!
		\left[\begin{array}{cccc}
			\!0  \!\\
			\!{1}/{m_2} \! \\
			\!0 \!\\
		    \!{-1}/{(m_2l)} \!\\
				\end{array}
		\right]\!Kx(t)
	\end{equation*}	
\\
\emph{where $m_1=1kg $ is the mass of pendulum bob,  $m_2=10kg$  is the cart mass, $l=3m$ is the length of the pendulum arm, and $g=10m/s^2$ is the gravitational acceleration. The state variables $[x_1,~ x_2,~ x_3, ~x_4]^T$ are the cart's position, the cart's velocity, the pendulum angle, and the pendulum bob's angle velocity. The system matrices are assumed unknown, but we have $\rho=50$
measurements $\{\dot x(T_i),~x(T_i),~u(T_i)\}^{50}_{i=1}$ with the interval satisfying $T_{i+1}-T_i=0.1$, where the input was sampled uniformly from $u(t)\in [-1,~1]$.
The disturbance $w(t)\in [-0.001,\,0.001]$
fulfills Assumption \ref{Ass:disturbance} with $Q_d=-I$, $S_d=0$, and $R_d=0.001^2\rho I$ $(\rho=100)$.}
\end{Example}
We set
the sampling interval $h=0.1$, triggering parameters $\sigma_1=0.01$, $\sigma_2=0.01$, $\theta=2$, $\lambda=0.5$, delay bounds $\underline{d}=0$, $\bar{d}=0.01$, and the matrix $B_w=0.01I$.
Solving the LMIs \eqref{Th:controller:AB:LMI1} and \eqref{Th:controller:AB:LMI2} in Theorem \ref{Th:Co-designing:AB} with $\epsilon=2$,
the controller gain and the triggering matrix were co-designed as follows
\begin{align*}
K&=[
  \begin{array}{cccc}
 2.4891  & 11.9277 & 297.9385 & 166.4980\\
  \end{array}
]\\
\Omega&=\left[
  \begin{array}{cccc}
  19.8622 & -49.6572  & -4.7734  & 18.2200\\
  -49.6572  &196.0430 &  11.0141 & -61.6847\\
   -4.7734 &  11.0141  &  1.3070  & -4.4237\\
   18.2200 & -61.6847  & -4.4237 &  20.7782\\
  \end{array}
\right].
\end{align*}
We simulated system \eqref{sys:sampling} under the triggering scheme in \eqref{sys:trigger} with the initial condition $x(0) = [0.98,~0, ~ 0.2,~0 ]^T$, as well as the dynamic variable $\eta(t)$, over the time interval $[0,\,30]$. The results in Fig. \ref{FIG:xt3}
show that $x(t)$  and  $\eta(\tau_j-d)$ approach zero as $t\to\infty$
under the designed controller gain $K$ and the triggering matrix $\Omega$.
In addition, under the triggering scheme  \eqref{sys:trigger}, \emph{only} $94$ out of 300 sampled data were transmitted to the controller, corroborating the effectiveness of  our data-driven event-triggering scheme on the fourth-order system.}
\section{Concluding Remarks}\label{sec:conclude}
In this paper, we proposed a general dynamic event-triggering transmission scheme for sample-data control systems with communication delays.
As an intermediate result of independent interest, we also provide a novel looped-functional-based approach for event-triggered control with multiple advantages (reduced computational complexity, reduced conservatism) over existing methods in the literature.
Both model- and data-based stability conditions were developed using the looped-functional approach.
Meanwhile, methods for co-designing the controller gain and the triggering matrix were presented for both cases of having a known or  an unknown input matrix.
Finally, two numerical examples were provided to corroborate the role of additional prior knowledge of the input matrix in reducing conservatism of our data-driven stability conditions, as well as the merits and effectiveness of our methods.

\appendices

\bibliographystyle{IEEEtran}
\bibliography{cas-refs}

\end{document}